\newif\ifproc
\proctrue 
\procfalse

\ifproc
  \documentclass{llncs}
\let\llncssubparagraph\subparagraph
\let\subparagraph\paragraph
\usepackage[compact]{titlesec}
\let\subparagraph\llncssubparagraph  
  
  \usepackage{environ}
  \NewEnviron{killcontents}{}

\else
  \documentclass[letterpaper,11pt]{article}
  \usepackage{fullpage}

  \usepackage{typearea}
  \typearea{14}
  \usepackage[compact]{titlesec}
  \usepackage{amsthm}
	\usepackage{algorithm}
	\usepackage{algpseudocode}
	\usepackage{float}


\fi

\newcommand{\full}[1]{\ifproc\else#1\fi}
\newcommand{\short}[1]{\ifproc#1\fi}

\usepackage{amsmath, amssymb}
\usepackage{color}
\usepackage{mathrsfs,bbm,bm}
\usepackage[pdftitle={Tile},
            pdfauthor={Author},
            pdfkeywords={Keywords}]{hyperref}
\usepackage{xspace}

\ifproc
  \newtheorem{thm}{Theorem}
\else
  \newtheorem{thm}{Theorem}[section]
\fi
\newtheorem{ithm}[thm]{Informal Theorem}
\newtheorem{cor}[thm]{Corollary}

\newtheorem{fact}[thm]{Fact}

\ifproc
\else
  \newtheorem{lemma}[thm]{Lemma}
  
  \newtheorem{claim}[thm]{Claim}
  \newtheorem{definition}[thm]{Definition}

\fi

\newcommand{\E}{\mathbb{E}}
\newcommand{\R}{\mathbb{R}}

\newcommand{\A}{\mathbf{A}}
\newcommand{\C}{\mathbf{C}}
\newcommand{\p}{\mathbf{p}}
\newcommand{\w}{\mathbf{w}}
\renewcommand{\o}{\mathbf{o}}
\renewcommand{\v}{\mathbf{v}}
\newcommand{\ho}{{\o^*}}
\newcommand{\ip}[2]{\left\langle #1, #2 \right\rangle}
\newcommand{\bip}[2]{\big\langle #1, #2 \big\rangle}
\newcommand{\eo}{\tilde{o}}

\newcommand{\bpi}{\boldsymbol{\pi}}
\newcommand{\ba}{\mathbf{a}}
\newcommand{\btau}{\boldsymbol{\tau}}
\newcommand{\tx}{\widetilde{x}}
\newcommand{\hx}{\widehat{x}}
\newcommand{\bx}{\mathbf{x}}
\newcommand{\by}{\mathbf{y}}
\newcommand{\bH}{\mathbf{H}}
\newcommand{\cE}{\mathcal{E}}

\newcommand{\bI}{\mathbf{I}}
\newcommand{\bJ}{\mathbf{J}}
\newcommand{\bK}{\mathbf{K}}
\newcommand{\bT}{\mathbf{T}}
\newcommand{\D}{\mathcal{D}}
\newcommand{\bY}{\mathbf{Y}} 

\newcommand{\bL}{\boldsymbol \L}
\newcommand{\bD}{\boldsymbol \D}

\newcommand{\balpha}{\boldsymbol \alpha}
\newcommand{\bbeta}{\boldsymbol \beta}

\newcommand{\fullsimplex}{\blacktriangle}
\newcommand{\simplex}{\bigtriangleup}

\newcommand{\smid}{\,|\,}

\let\originalleft\left
\let\originalright\right
\renewcommand{\left}{\mathopen{}\mathclose\bgroup\originalleft}
\renewcommand{\right}{\aftergroup\egroup\originalright}

\DeclareMathOperator{\Var}{Var}

\newcounter{mynote}[section]

\newcommand{\new}[1]{#1}

\newcommand{\remove}[1]{}

\renewcommand{\L}{\mathcal{L}}
\newcommand{\e}{\varepsilon}
\newcommand{\ones}{\mathbbm{1}}
\newcommand{\RHS}{\mathsf{rhs}}
\newcommand{\OPT}{\ensuremath{\mathsf{opt}}\xspace}

\newcommand{\loadBal}{\textsf{Load-balance}\xspace}
\newcommand{\PCMC}{\textsf{PCMC}\xspace}
\newcommand{\PLP}{\textsf{PLP}\xspace}
\newcommand{\LPtoLB}{\textbf{LPviaLB}\xspace}
\newcommand{\expertLB}{\textbf{expertLB}\xspace}
\newcommand{\MW}{\textbf{MWU}\xspace}
\newcommand{\tsty}{\textstyle}

\newcommand{\cG}{\mathcal{G}}
\newcommand{\OTL}{\textbf{OTL}\xspace}
\newcommand{\DLA}{\textbf{DLA}\xspace}
\newcommand{\mOTL}{\textbf{mOTL}\xspace}
\newcommand{\mDLA}{\textbf{mDLA}\xspace}

\newcommand{\regA}{\alpha}

\newcommand{\optp}{{p^*}}
\newcommand{\boptp}{{\p^*}}
\newcommand{\bopto}{{\o^*}}
\newcommand{\Et}{{\mathbb{E}}}

\newcommand{\initOneLiners}{%
    \setlength{\itemsep}{0pt}
    \setlength{\parsep }{0pt}
    \setlength{\topsep }{0pt}
}
\newenvironment{OneLiners}[1][\ensuremath{\bullet}]
    {\begin{list}
        {#1}
        {\initOneLiners}}
    {\end{list}}

\title{How the Experts Algorithm Can Help Solve LPs Online\footnote{A preliminary version of this work appeared in ESA 2014.}}
\date{\today}

\ifproc

 \author{Anupam Gupta\inst{1}\thanks{Research partly supported by
    NSF awards CCF-1016799 and CCF-1319811, and by a grant from the
    CMU-Microsoft Center for Computational Thinking.} \and Marco Molinaro\inst{2}}

 \institute{Computer Science Department, Carnegie Mellon
  University, Pittsburgh, PA. \and ISyE, Georgia Tech, Atlanta, GA.}

\else
 \author{Anupam Gupta\thanks{Computer Science Department, Carnegie Mellon
  University, Pittsburgh, PA. Research partly supported by
    NSF awards CCF-1016799 and CCF-1319811, and by a grant from the
    CMU-Microsoft Center for Computational Thinking.} \and Marco Molinaro\thanks{EWI, TU Delft.}}

\fi

\begin{document}
\maketitle

\begin{abstract}
  We consider the problem of solving packing/covering LPs online, when
  the columns of the constraint matrix are presented in random
  order. This problem has received much attention and the main focus is to figure out how large the right-hand sides of the LPs
  have to be (compared to the entries on the left-hand side of the
  constraints) to allow $(1+\e)$-approximations online. It is known that
  the right-hand sides have to be $\Omega(\e^{-2} \log m)$ times the left-hand sides,
  where $m$ is the number of constraints. 

  In this paper we give a primal-dual algorithm that achieve this bound for 
  mixed packing/covering LPs. Our
  algorithms construct dual solutions using a
  regret-minimizing online learning algorithm in a black-box fashion, and use them to construct primal solutions. The adversarial guarantee that holds for the constructed duals helps us to take care of most of the correlations that arise in the algorithm; the remaining correlations are handled via martingale concentration and maximal inequalities. These ideas lead to conceptually simple and modular algorithms, which
  we hope will be useful in other contexts.
\end{abstract}

\section{Introduction}

	Sequential decision-making is one of the most important topics in Operations Research. While several sequential decision models are available, the common feature is that items/requests come over time and at a given time only partial information about the future items/requests is known. Decisions need to be made over time only with the information available at that point. Moreover, typically decisions are irrevocable.
	
While such problems have been extensively studied in the worst-case competitive analysis \cite{buchbinderBook,packingCovering13}, in order to avoid the pessimistic achievable guarantees a lot of recent research has focused on the \emph{random-order} model, which can be broadly described as follows: while the underlying instance is arbitrary (i.e. chosen adversarially), the items/requests come in a \textbf{random order} (i.e. the adversary cannot control the order in which the instance is presented). Several sequential decision problems have been considered in this model, starting from a classic maximization version of the secretary problem \cite{Dynkin}, to single-knapsack problems \cite{kleinberg,babaioff}, the AdWords problem \cite{MSVV,GoelMehta08,DevanurHayes09}, matroid secretary problems~\cite{BIK07, BabaioffSurvey}, and packing LPs \cite{feldman,AWY14,MR12,KRTV13,AgrawalD15,UW}. 
	
	In this paper we consider a broad class of such problems, namely that of solving online packing/covering multiple-choice LPs; this class includes most of the examples above. 

	
	\subsection{Model and prior results} \label{sec:introModel}
	
	The offline version of the \emph{packing/covering multiple-choice LPs} we consider have the form
	\begin{align*}
		\max ~\tsty \sum_{t = 1}^n \pi^t x^t & \tag{\PCMC} \\
		  st ~\tsty \sum_{t = 1}^n A^t x^t &\le b \\
		     \tsty \sum_{t = 1}^n C^t x^t &\ge d \\
		     x^t &\in \fullsimplex^k \ \ \ \forall t \in [n]
	\end{align*}
	where all the data $\pi^t, A^t, C^t, b, d$ is \textbf{non-negative} and $\fullsimplex^k$ denotes the ``full simplex'' $\{x \in [0,1]^k : \sum_j x_j \le 1\}$. Note that the variables $\{x^t_1, x^t_2, \ldots, x^t_k\}$ in each block $t$ are bound together by $\sum_j x^t_j \leq 1$ --- this is the ``multiple-choice'' setting, which is useful in many contexts. In the case $k=1$, the multiple-choice model degenerates to the
standard online packing-covering LP model where a single variable
$x^t$ is revealed at each time, and must be assigned a fractional value in $[0,1]$.
	
	In the \emph{random-order model}, the number $n$ of columns  and the right-hand sides $b$ and $d$ are known a priori but the columns of the LP are revealed one by one in random order, that is, their arrival order is chosen uniformly at random out of all possible permutations. When a column-block $(\pi^t, A^t, C^t)$ is revealed, the corresponding $x^t$ variables need to be irrevocably set before the next column-blocks are revealed. The goal is to obtain a feasible solution for the LP that tries to obtain as much value (in expectation or with high probability) as the optimal \textbf{offline} solution.	This is exactly the model we consider here.
	
	Packing/covering multiple-choice LPs form a broad class of problems. Special cases of this model with \textbf{packing-only} constraints (i.e. without the constraints $\sum_{t = 1}^n C^t x^t \ge d$) have been extensively studied \cite{Dynkin,kleinberg,babaioff,MSVV,GoelMehta08,DevanurHayes09,feldman,AWY14,MR12,KRTV13,AgrawalD15,UW}. These models have a vast range of applications, like online advertisement, online routing, and airline revenue management. 
	
	Most of these results show that one can obtain online an $(1-\e)$-approximation for the LP as long as the right-hand side of the constraints are sufficiently (depending on $\e$) larger than any entry in the left-hand side matrix. 
\new{This direction started with the work of Kleinberg~\cite{kleinberg}, who considered the case without multiple choices (i.e. $k = 1$) and with a single cardinality constraint $\sum_t x^t \le b$. 
        He showed a $(1-\e)$-approximation in expectation as long as $\frac{b}{\max_t a^t} \ge \Omega(\epsilon^{-2})$, which is best possible.} 
        Subsequently, Devanur and Hayes considered the AdWords problem~\cite{DevanurHayes09}, which was then extended to general packing multiple-choice LPs \cite{feldman,AWY14,MR12}. These results provide an $(1-\e)$-approximation as long as the right-hand sides/item sizes ratio satisfies $\frac{\min_i b_i}{\max_{t,i,j} A^t_{ij}} \ge \Omega(\e^{-2})\, \min\{m^2 \log (m/\e), m \log (n/\e)\}$, where $m$ is the number of packing constraints and $n$ the number of column-blocks. Agrawal et al.~\cite{AWY14} also showed that $\Omega(\e^{-2} \log m)$ ratio  is required for $(1-\e)$-approximations, which was believed to be the right answer.

	\new{However, despite the increasing interest of solving LPs online, to the best of our knowledge no guarantee is known for online optimization of LPs with \textbf{covering} constraints, even in the weaker i.i.d. model (see \cite{MR12} for discussion of random order vs. i.i.d.)\footnote{\cite{Devanur11} presents an online for the \textbf{feasibility} version of packing-covering multi-choice LPs in the i.i.d. model.}} 

\subsection{Our results}
\label{sec:our-results}

\short{Let us now discuss our results in more detail. For brevity,
several proofs are omitted from the paper and are deferred to its full version.
}

\paragraph{Packing/covering multiple-choice LPs.} 


\new{Our main contribution is the first set of algorithms with provable guarantees for packing/covering multiple-choice LPs in the random-order model. Our algorithm handles general packing/covering multi-choice LPs and achieves the optimal guarantee even when only packing constraints are present.}
%
For example, in case an estimate of the optimal value is present we have the following guarantees.

\begin{ithm}[Packing/covering LPs]
  \label{ithm:main-I}
  Consider a feasible packing/covering multiple-choice LP in the random-order model. Given $\e, \delta$, assume that each right-hand side in the LP is at least $\Omega(\frac{\log (m/\delta)}{\e^2})$ times the maximum coefficient of the left-hand side in its row, and that the optimal value is at least $\Omega(\frac{\log (m/\delta)}{\e^2})$ times each individual item's profit. Then given an (under)estimate $O^*$ for the optimal value of the offline LP, the algorithm \LPtoLB computes an $\e$-feasible solution online that achieves a value of $(1-\e)\,O^*$ with probability $1 - \delta$.
\end{ithm}

(For the formal version, see Theorem~\ref{thm:LPAssEstimates}.)
Moreover, the algorithm is very efficient and does not need to solve any (offline) Linear Programs. Also, as is common in this context, our algorithm produces \emph{integer} solutions that compare favorably to the optimal offline fractional solution. 

The assumption of the availability of an estimate for the optimal value is often a reasonable one (e.g.\ when enough historic knowledge of the problem directly provides such an estimate). Nonetheless, we show how to construct such an estimate as the columns arrive online. However, to obtain good estimates we need to consider \emph{stable} LPs, which informally means that allowing solutions to violate the covering constraints by a small factor does not increase the optimal value by much (see Definition \ref{def:stability}, and also Lemma~\ref{lemma:sufStable} for a simple sufficient condition for stability). We note that a packing-only LP is vacuously stable. 

\begin{ithm}[Packing/covering LPs II]
  \label{ithm:main-II}
  Consider a feasible, stable, packing/covering multiple-choice LP in the random-order setting. Given $\e, \delta$, assume that each right-hand side in the LP is at least $\Omega(\frac{\log (m/\delta)}{\e^2})$ times the maximum coefficient of the left-hand side in its row, and that the optimal value is at least $\Omega(\frac{\log (m/\delta)}{\e^2})$ times each individual item's profit. Then with probability at least $1-\delta$ the algorithm \DLA computes an $\e$-feasible solution online that has value at least $(1 - \e)$ times the optimal value of the offline LP.
\end{ithm}

(The formal version appears as Theorem~\ref{thm:multi-phasePC}.) The high-level idea to obtain an estimate of the optimum is reasonably standard at this point (\cite{kleinberg,Devanur11,AWY14,MR12,UW}: one solves an LP with the columns seen so far to obtain an estimate, which is updated at doubling times. However, the presence of covering constraints introduces additional technical difficulties, and we need the assumption of stability of the LP to control the variance of estimate produced (this issue is further discussed in \S\ref{sec:est-opt}).



\paragraph{Packing LPs.} \short{\emph{This entire section appears only in the full version of the paper.
}} Informal Theorem~\ref{ithm:main-II} assumes that individual item profits are
small compared to the optimal value. While this is often reasonable, can we remove this assumption? Our final result removes all assumptions about the magnitude of
the profit whenever we are dealing with packing-only LPs (and without
multiple-choice constraints of the form $\sum_j x^t_j \leq 1$). 
Obtaining good estimates of $\OPT$ without any assumptions on the item values in our setting requires significantly new ideas (also discussed further below in \S\ref{sec:est-opt}). (The formal version of the following theorem appears as Theorem~\ref{thm:PLP-main}.)

\begin{ithm}[Ultimate packing LP]
  \label{ithm:packing}
  Consider a packing LP in the random-order setting. Given $\e > 0$,
  suppose the right-hand sides are $\Omega(\frac{\log (m/\e)}{\e^2})$
  times any left-hand side entry. Then the algorithm \mDLA
   computes a solution online with expected value at least $(1
  - \e)$ times the optimal value of the offline LP.
\end{ithm}

	\paragraph{Generalized load balancing.}
	
		The main building block for our algorithms for packing/covering LPs is the ability to solve the following \emph{generalized load-balancing problem}: We have a set of $m$ machines. For each arriving job $t$, there are $k$ different options on how to serve it, and the $j^{th}$ choice
requires some amount of processing on each of the machines, given by the
vector $(a^t_{1j}, a^t_{2j}, \ldots, a^t_{mj})$.  When the job arrives,
we must (fractionally) choose one of these options. The goal is to
minimize the makespan, i.e., the maximum total processing assigned to any
machine.
Notice that when the matrices $A^t$'s are diagonal $m \times m$, the $j^{th}$ choice places $a^t_{jj}$ load only on the $j^{th}$ machine, and hence captures the classic problem of scheduling on unrelated machines. 


\begin{ithm}[Load-balancing]
  \label{ithm:lb}
  Consider the load-balancing problem where the jobs arrive in random
  order. Given $\e > 0$, if the optimal makespan $\lambda^*$ is at least
  $\Omega(\frac{\log (m/\delta)}{\e^2})$ times the maximum amount of
  processing required by any option for any job, the algorithm \expertLB constructs an online solution with makespan at most $(1+\e)\lambda^*$ with probability $1 - \delta$.
\end{ithm}

\new{(The formal version in Theorem~\ref{thm:loadBalHP} proves this with respect to a generic low-regret online learning algorithm for the prediction-with-experts problem. For readability, here we instantiate it with parameters from the Multiplicative Weights Update algorithm.)} The case where the jobs are \textbf{i.i.d.}\ draws from some unknown distribution was considered in~\cite[Section~3]{Devanur11}, where with probability $1- \delta$, a $(1+\e)$-approximation to
the optimum makespan is obtained assuming the optimum is $\Omega(\frac{\log
  (m/\delta)}{\e^2})$ times the maximum processing $a^t_{ij}$. Our theorem above shows that the same guarantee holds in the more general random-order model, answering an open question from \cite{Devanur11}.



	In fact, we can handle processing requirements that can
be both positive or negative, as long as for each machine the jobs are
mostly positive or mostly negative (see Definition \ref{def:well-bounded}). \new{Handling negative processing requirements is crucial for addressing covering constraints when we \textbf{reduce} packing/covering LPs to the generalized load-balancing problem. We note that \cite{Devanur11} also considered this generalized load-balancing problem in connection with packing/covering LPs; however, they do not consider negative processing times and thus do not provide a direct reduction between the problems.}


\subsection{The Road-map}

\new{Our algorithmic development is modular: before elaborating on our technique, let us lay out the plan of attack.
\begin{OneLiners}
\item In \S\ref{sec:lb} we show our results for generalized load-balancing on random-order arrival streams via online learning algorithms.
\item In \S\ref{sec:lp-to-lb} we give a simple reduction that reduces solving packing/covering LPs \emph{with an estimate of the optimal value} to generalized load-balancing, and hence prove Theorem~\ref{ithm:main-I}.
\item In \S\ref{sec:estPC} we show how to estimate the optimal value for packing/covering LPs.
\item In \S\ref{sec:random-pclps}, we combine ideas from \S\ref{sec:lp-to-lb} and \S\ref{sec:estPC} to prove Theorem~\ref{ithm:main-II}, our main result about packing/covering LPs.
\item Finally, in \S\ref{sec:plp-no-limit} we refine the results of \S\ref{sec:random-pclps} for packing-only LPs: we remove the requirement that the individual item profits should be small, and get Theorem~\ref{ithm:packing}.
\end{OneLiners}
}



\subsection{Our Techniques}

	\paragraph{Decoupling primal and dual via online learning.}


	Dual-based method are very important in optimization under uncertainty and have been used extensively in the Operations Management literature \cite{simpson,elWilliamson,Talluri1998}. \new{One advantage of primal-dual methods is that they typically have fast update times, and avoid solving an LP at every step of the process. In addition, the computed duals can be used as a good heuristic to estimate the importance of each constraint and to perform sensitivity analysis.}
	
	Most of the works on packing LPs in the random-order model \cite{DevanurHayes09,feldman,AWY14,MR12} use the following form of the primal-dual approach (the primal-only algorithm of \cite{KRTV13} being a notable exception)\footnote{Primal-dual methods of a slightly different flavor are also extensively used in the worst-case model of online algorithms \cite{buchbinderBook}.}: a good dual LP solution is estimated based on the columns seen so far, then the reduced costs defined by this dual are used to decide how to set the primal variables. While being a natural idea, analyzing it poses the following crucial difficulty: the primal and dual solutions, by construction, are correlated in a non-trivial manner. In order to deal with these correlations, all of the above analyses resort to a massive union bound, which is the root of the extra $\log n$ and $m$ factors over the right ratio bound $\Omega(\epsilon^{-2} \log m)$.

	One of our contributions is an alternative algorithmic construction where the primal and the dual are only loosely coupled. For that, we construct duals using a \new{generic} \emph{regret-minimizing online learning} algorithm (which offers \emph{adversarial} guarantees); in the primal, we use a greedy strategy which ensures that, with respect to these duals, we are doing at least as well as the optimal offline solution. 
	
	The fact that the guarantees for the dual hold in the adversarial setting, together with the greedy primal step, reduces the analysis of the algorithm to that of comparing the dual solution constructed to the \emph{offline optimal} solution (see \S \ref{sec:guar-expect}). Since the latter is a \emph{fixed} solution (modulo the permutation of the columns), it is only loosely correlated with our constructed dual; in fact, correlations only arise because in the random-order model we are sampling \emph{without} replacement (in particular, in the weaker i.i.d. model there are no such correlations, see \S \ref{sec:lb} up to equation \eqref{eq:guaranteeLB}).
	
	
	 As far as we know, this is the first explicit black-box connection between online learning and optimization in the random-order model (see \S \ref{sec:indepWork} for independent and subsequent work); in the offline packing/covering LP setting such a connection was implicit in, e.g., \cite{tardos,young,gargKonemann} and made explicit in, e.g., \cite{K04,hazan,AHK12,CHW}. In the simpler i.i.d.\ setting, \cite{Devanur11} uses multiplicative-weights based dual updates. This is reminiscent of the Multiplicative Weight Update algorithm for online learning. Our analysis abstracts out the reliance on the details of that technique, giving a clean conceptual explanation of what is happening. \new{It also allows us to use \emph{any} experts online learning algorithm in our reductions, where the quality of the LP solution depends on the parameters in the rate of regret.}  Importantly, this isolation between the analysis of the online learning algorithm and the comparison of our dual solution and the offline optimum is instrumental in dealing with the dependencies arising in the random-order model. 
	
	\paragraph{Handling dependencies in random permutations.} Despite the loose coupling of primal and dual solutions, the correlations arising from sampling \emph{without} replacement cannot be immediately handled:
		the known bounds on the statistical distances between sampling with and without replacement (e.g.~\cite{diaconis1980}) do not seem to be strong enough, and a simple approach requires taking a union bound across all time steps, leading again to extra $\log n$ factors. To overcome this, we instead use a \emph{maximal inequality} for sampling without replacement to compare sampling with/without replacement at \emph{every} time step while avoiding a union bound, together with martingale concentration inequalities. The arguments we develop are quite general and may find use in other problems in the random-order model. 
		
	\new{
	\subsubsection{Estimating OPT: Challenges and Ideas} 
        \label{sec:est-opt}

In our approach we need to estimate (with high probability) the optimal value $\OPT$ of the offline LP. \emph{Under the assumption that the individual item profits are small compared to $\OPT$}, such estimate can be obtained by solving the LP comprising the first few columns seen in random order, and scaling its optimal value appropriately (actually for packing/covering LPs even this case is not that easy, and we require stability from the LP to provide guarantees for such estimates). However, without assuming anything on the individual item sizes, as we do for the packing-only case in \S\ref{sec:plp-no-limit}, the high-valued items make this estimate have too high a variance and precludes the high-probability guarantees required. In fact \textbf{all} previous approaches that use explicit estimates of $\OPT$ make the additional assumption that each item profit is at most $O(\frac{\e^2}{\log m})\,\OPT$ \cite{Devanur11,UW,AgrawalD15}\footnote{The requirement in \cite{AgrawalD15} comes from the use of Lemma 14 from \cite{Devanur11}, which uses this assumption.} To remove such assumptions, we show that one can simply pick all items with sufficiently high value and add them to an $(1-\e)$-approximation for the remaining LP (which then does not contain high-valued items) and this only incur a small loss in total profit (see \S \ref{sec:plp-no-limit}, in particular Lemma \ref{lemma:pickTop}). Unfortunately we also need to estimate the \emph{threshold} of what constitutes high enough value; this is however an easier task, since we only need weak guarantees for this estimate. This estimate is obtained by again looking at the LP comprising the first few columns seen in random order but \emph{removing the $K$ highest valued items} in this sample; this provides sufficient control on the remaining item values for the analysis to go through.}

%

\new{
\subsection{Independent and subsequent work} \label{sec:indepWork}

Independent to our work, Kesselheim et al.~\cite{KRTV13} present a primal algorithm that achieves the first optimal guarantee for \textbf{packing-only} multiple-choice LPs. More precisely, if $d$ is the \emph{column-sparsity} of the LP (i.e., maximum number of non-zero entries in each column), they obtain an $(1-\e)$-approximation as long as the right-hand sides/item sizes ratios are at least $\Omega(\epsilon^{-2} \log d)$. The dependence on the column-sparsity $d$ rather than the number of constraints $m$ (and the fact they handle multi-choice constraints) makes their result quantitatively stronger than our Theorem \ref{ithm:packing}, while both results coincide and are optimal for $d = \Omega(m)$~\cite{AWY14}. However, our algorithm is more efficient since it requires us to solve only a small number of offline LPs, while their algorithm solves an LP at each of the $n$ timesteps (in the experiments presented in \cite{UW}, the algorithm from \cite{KRTV13} is 2 to 3 orders of magnitude slower than algorithms that solve roughly as many LPs as ours does).

Subsequent to our work, Eghbali et al.~\cite{UW} provided a primal-dual algorithm for solving more general packing-type \emph{convex} programs. They also propose a primal-dual algorithm that draws inspiration from \cite{Devanur11}; in particular, they also use a sort of exponential update of duals. In their analysis, they also employ martingale methods and a (Doob's) maximal inequality.

Also subsequent to our work, Agrawal and Devanur \cite{AgrawalD15} consider an even broader class of online convex programs (see the paper for proper description). Their primal-dual algorithm extends the one we present in \S \ref{sec:lb} for the load-balancing problem and generalizes our connection between online (convex) learning and i.i.d./random-order convex optimization. A key element for this generalization is the use of \emph{Fenchel duality}.
One drawback is that their general result only satisfies the constraints \emph{in expectation}. However, for the case of packing-only LPs their techniques are able to give bona fide feasible solutions also with optimal guarantees. 





Finally, independently Agrawal and Devanur \cite{agrawalEC14} used a similar technique to devise efficient an algorithm for the Bandits with concave rewards and concave knapsacks problem (Algorithm 3), again using Fenchel duality. They trace back the inspiration for this approach to \cite{hazanBlackwell}, which show how solve the Blackwell approachability problem (which at a very high-level asks for a sequence of iterates that approaches a convex set) using online convex optimization.

}



\section{Preliminaries}


For a vector $v \in \R^k$, define its max-norm $\|v\|_{\max} := \max_j v_j$; notice this does not involve absolute values as in the definition of the $\ell_\infty$ norm. Also, we use $|v| := ( |v_1|, |v_2|, \dots, |v_k|)$ to denote the component-wise absolute value of a vector.
Recall that  $\Delta^k$ denotes the simplex $\{ p \in [0,1]^k : \sum_i p_i =
1\}$, and $\fullsimplex^k$ denotes the ``full simplex'' $\{x \in [0,1]^k
: \sum_j x_j \le 1\}$.

		\subsection{Prediction-with-experts problem}
                \label{sec:experts}
		
		As mentioned before, our primal-dual algorithms will have their dual solutions constructed via an algorithm for maximization version of the classic \emph{prediction-with-experts problem} \cite{CBL06}. In this online problem, \textbf{adversarial} vectors $o^1, o^2, \ldots, o^n \in [-M,M]^m$ are presented one by one to the algorithm. At time $t$, using only the information $o^1, \ldots, o^{t-1}$ up to time $t-1$, the algorithm needs to compute a vector $w^t$ in the simplex $\Delta^m$; it then incurs a reward of $\ip{w^t}{o^t}$. The goal is to maximize the total reward obtained $\sum_{t} \ip{w^t}{o^t}$. (Note we have stated this problem with the domain of the adversarial vectors being the scaled one $[-M, M]^m$ instead of $[-1,1]^m$.)
		
	The quality of the solution constructed by the algorithm is measured using the notion of \emph{regret}, which compares against the best \emph{fixed} $w \in \Delta^m$ that knows all the vectors $o^t$ in hindsight, namely $\max_{w \in \Delta} \sum_t \ip{w}{o^t} = \|\sum_t o^t\|_{\max}$. Here we consider a slightly different parametrization of regret, which is also used in \cite{AHK12,KalaiVempala}.
	
	\begin{definition}[Regret]
		We say that an algorithm for the prediction-with-experts problem has \emph{$(\regA, R)$-regret} if the solution $w^1, \ldots, w^n$ constructed satisfies
		\begin{align*}
			\sum_{t = 1}^n \ip{w^t}{o^t} \ge \left\|\sum_{t=1}^n o^t - \alpha \sum_{t = 1}^n \left|o^t\right| \right\|_{\max} - R.
		\end{align*}
	\end{definition}
		
	This problem
	 has been extensively studied and many algorithms achieving optimal regret are known \cite{Shalev-ShwartzBook}. The traditional $(0, M \sqrt{n \log m})$-regret is attained by several algorithms and is known to be optimal if the first parameter $\alpha$ in the regret is kept at 0. However, the guarantee that works best for our purposes is attained by the well-known \emph{Multiplicative Weights Update} algorithm, which gives an $(\regA, R)$-regret with (typically) much smaller value of $R$ at the expense of a non-zero $\regA$.
	
	\begin{thm}[Theorem 2.5 of \cite{AHK12}]\label{thm:MW}
		For every $\epsilon \in [0,1]$, the Multiplicative Weights Update (\MW) algorithm with learning parameter $\epsilon$ achieves $(\epsilon, \frac{M \log m}{\epsilon})$-regret.
	\end{thm}
	

	\subsection{Concentration inequalities}
        \label{sec:conc-ineq}
	
	Another important tool for our analyses are strong concentration inequalities for non-independent processes. First we need a version of Bernstein's inequality for sampling without replacement. 

\begin{thm}[Theorem 2.14.19 in
  \cite{weakConvergence}] \label{thm:conc-hyp} Let $Y = \{Y_1, \ldots,
  Y_n\}$ be a set of real numbers in the interval $[0,M]$. Let $S$ be a
  random subset of $Y$ of size $s$ and let $Y_S = \sum_{i \in S}
  Y_i$. Setting $\mu = \frac{1}{n} \sum_i Y_i$ and $\sigma^2 =
  \frac{1}{n}\sum_i (Y_i - \mu)^2$, we have that for every $\tau > 0$
  \begin{align*}
    \Pr(|Y_S - s \mu| \ge \tau) \le 2 \exp \left( -\frac{\tau^2}{2 s
        \sigma^2 + \tau M} \right)
  \end{align*}
\end{thm}

	Notice that we can use the fact $Y_i \in [0,M]$ to upper bound the variance $\sigma^2$ as $$\sigma^2 \le
\frac{M}{n} \sum_i |Y_i - \mu| \le \frac{M}{n} \left( \sum_i |Y_i| +
  \sum_i |\mu| \right) = 2 M \mu,$$ which directly gives the following. 
	
	\begin{cor}\label{cor:multiChernoff} Let $Y = \{Y_1, \ldots, Y_n\}$ be
  a set of real numbers in the interval $[0,M]$. Let $S$ be a random
  subset of $Y$ of size $s$ and let $Y_S = \sum_{i \in S} Y_i$. Setting
  $\mu = \frac{1}{n} \sum_i Y_i$, we have that for every $\tau > 0$,
  \begin{align*}
    \Pr(|Y_S - s \mu| \ge \tau) \le 2 \exp \left( -\frac{\tau^2}{M(4 s \mu + \tau)} \right) \le 2 \exp \left(-\min\left\{\frac{\tau^2}{8 M s \mu}, \frac{\tau}{2M}\right\}\right).
  \end{align*}
	\end{cor}
		
	We also need a \emph{maximal} Bernstein's inequality for sampling without replacement. While we could not find such inequality explicitly stated in the literature, it easily follows from Theorem \ref{thm:conc-hyp} and a Levy-type maximal inequality for exchangeable random variables of \cite{pruss}; its proof is presented in Appendix \ref{app:maximalBernstein}.

	\begin{lemma} \label{lemma:maximalBernstein}
		Consider a set of real values $x_1, \ldots, x_n$ in $[0,1]$, and let $X_1, \ldots, X_k$ be sampled without replacement from this collection. Assume $k \le n/2$. Let $S_i = X_1 + \ldots X_i$. Also let $\mu = \frac{1}{n} \sum_i x_i$ and $\sigma^2 = \frac{1}{n} \sum_i (x_i - \mu)^2$. Then for every $\alpha > 0$
		\begin{align*}
			\Pr\left(\max_{i \le k} |S_i - i \mu| \ge \alpha\right) \le 30 \exp\left(\frac{(\alpha/24)^2}{2 k \sigma^2 + (\alpha/24)} \right).
		\end{align*}
	\end{lemma}

	Finally, we need the following classic martingle concentration inequality.
	
	\begin{thm}[Freedman's Inequality \cite{freedman}] \label{thm:freedman}
		Let $X_0, X_1, \ldots, X_n$ be a real-valued martingale, with $X_0 = 0$, with respect to a filtration $\mathcal{F}_0 \subseteq \mathcal{F}_1 \subseteq \ldots \subseteq \mathcal{F}_n$. Let $Y_t = X_t - X_{t-1}$ be the martingale differences, let $M$ be such that $|Y_t| \le M$ for all $t$ and let $L = \sum_{i \le n} \Var(Y_t \smid \mathcal{F}_{i-1})$ be the predictable variance. Then for every real numbers $\alpha, \sigma^2 \ge 0$,
		\begin{align*}
			\Pr\left( X_n \ge \alpha \textrm{ and } L \le \sigma^2\right) \le \exp\left(- \frac{\alpha^2}{2 \sigma^2 + 2 M \alpha} \right)
		\end{align*}	
	\end{thm}


	\section{Load-balancing using experts}
        \label{sec:lb}	

        In this section we formally define the generalized
        load-balancing problem and present our online algorithm for it.
        In the next section, we will use this algorithm to solve
        packing/covering LPs given an estimate of the optimal value.


        \subsection{Definitions: offline and online instances}
        \label{sec:def-lb}

	An instance of the \emph{offline} version of the \emph{generalized load-balancing problem} is a set of matrices $\{A^t\}_t$, each in $\R^{m \times
          k}$. The goal is to find vectors $p^1, p^2, \ldots, p^n \in
        \Delta^k$ to minimize $\|\sum_{t =1}^n A^t p^t\|_{\max}$, the
        load of the most-loaded machine.
        We use  $\lambda^*$ throughout to denote the offline optimum value. Notice that, by adding a  column of zeroes, we can get an
  instance where one option is to do nothing, i.e., we can change the problem
  to minimizing $\|\sum_{t =1}^n A^t p^t\|_{\max}$ over $p^t \in
  \fullsimplex^{k}$. 
	
	In the \emph{online} version of this problem, we consider the
        random permutation model. Now the number of time steps $n$ is
        the only information known upfront. Let $\A^1, \A^2, \ldots,
        \A^n$ be matrices sampled from the set $\{A^1, \ldots, A^n\}$
        uniformly \emph{without replacement}. At time step $t$, the
        algorithm has seen matrices $\A^1, \dots, \A^t$ and must decide
        on a (random) vector $\p^t \in \Delta$; the randomness is both due to the random matrices seen so far and the internal coin flips of the
        algorithm (if any).  The goal is to minimize $\|\sum_{t = 1}^n
        \A^t \p^t\|_{\max}$. The vectors $\{\p^t\}_t$ output by the
        algorithm are called the \emph{online solution} for the online
        instance $\{\A^t\}_t$ corresponding to the offline instance $\{A^t\}_t$.

\full{        Ideally we want guarantees that for any given $\e, \delta > 0$, we get $\|\sum_{t = 1}^n \A^t \p^t\|_{\max} \le (1+\e)\lambda^*$ with probability $(1- \delta)$. Due to the online nature of the problem, such a guarantee is only conceivable if $\lambda^*$ is ``large'' compared to the maximum size of the loads. In the rest of this section we show how to get the best possible such guarantee.
}

        \subsection{Well-bounded instances}

        While instances arising from scheduling-type applications usually consist solely of positive loads, our reduction used for packing/covering LPs generates instances where some entries $A^t_{ij}$ are also negative. In order to get good guarantees, we need to control how negative these entries can be. Loosely speaking, while the entries of the load matrices $A^t$ can be in the symmetric interval $[-M, M]$, for each machine we require that either the loads are mostly positive on all time steps (with any negative loads being tiny), or they are mostly negative on all time steps. 

		\begin{definition}[Well-bounded instance] \label{def:well-bounded}
		 For $M, \gamma \geq 0$, an instance $A^1, \ldots, A^n$ of the load-balancing problem is \emph{$(M,\gamma)$-well-bounded} if $A_{ij}^t \in [-M, M]$ for all $i,j,t$, and moreover for each $i \in [m]$ we have: either $A_{ij}^t \ge -\frac{\gamma\lambda^*}{n}$ for all $j$ and $t \in [n]$, or $A_{ij}^t \le \frac{\gamma\lambda^*}{n}$ for all $j$ and $t \in [n]$.
		\end{definition}
		
		In particular, this is satisfied with $\gamma = 0$ if the $A^t$'s are non-negative. The reader can think of $\gamma$ as a small constant, say one. The main motivation behind this definition is that it allows us to control the variation of random processes defined over $\{A^t\}_t$ in a way similar to how we obtained Corollary \ref{cor:multiChernoff}. 
		
		\begin{lemma} \label{lemma:wellBounded}
			Suppose $\{A^t\}_{t = 1}^n$ is an
                        $(M,\gamma)$-well-bounded instance for some $M,
                        \gamma \geq 0$, and consider  $\bar{p}^1, \ldots, \bar{p}^n \in \Delta^k$. Let the sequence $\A^1 \bar{\p}^1, \ldots \A^n \bar{\p}^n$ be sampled without replacement from the set $\{A^t \bar{p}^t\}_t$. Then for every event $\cE$ and for every $i \in [m]$ and $t \in [n]$, 
			$$\E\bigg[ \left|\A^t_i \bar{\p}^t - \E[\A^t_i \bar{\p}^t \smid \cE]\right| \;\bigg\rvert\;  \cE\bigg] \le \frac{2\gamma\lambda^*}{n} + 2 \left|\E\left[\A^t_i \bar{\p}^t \smid \cE \right]\right|.$$
		\end{lemma}
		
		\begin{proof}
			Fix $i \in [m]$. Using the well-boundedness condition, define $\bar{\v}^t_i$ as follows: if $A^t_i \bar{p}^t \ge -\frac{\gamma\lambda^*}{n}$ for all $t$, set $\bar{\v}^t_i = \A^t_i \bar{\p}^t + \frac{\gamma\lambda^*}{n}$ for all $t$; else set $\bar{\v}_i^t = \frac{\gamma\lambda^*}{n} - \A^t_i \bar{\p}^t$. By definition for all $t$ we have $\bar{\v}^t_i \ge 0$, and since we did a uniform shift
			\begin{align*}
				\E\Big[ \left|\A^t_i \bar{\p}^t - \E[\A^t_i \bar{\p}^t \smid \cE]\right| \;\Big\rvert\;  \cE\Big] = \E\Big[\left|\bar{\v}^t_i - \E[\bar{\v}^t_i \smid \cE]\right| \;\Big\rvert\; \cE\Big] \le \E\left[|\bar{\v}^t_i| \;\middle|\; \cE\right] + \left|\E\left[\bar{\v}^t_i \smid \cE\right]\right| = 2 \E\left[\bar{\v}^t_i \smid \cE\right].
			\end{align*}
		But it is easy to see that $\E[\bar{\v}^t_i \smid \cE] \le \frac{\gamma\lambda^*}{n} + \left|\E[\A^t_i \bar{\p}^t \smid \cE]\right|$, which concludes the proof.
		\end{proof}
		
		
	\subsection{The \expertLB algorithm and its guarantee}
		
		Given an online instance $\{\A^t\}_t$ to the generalized load-balancing problem and values $n, M$, and $\e$,
the following algorithm \expertLB (for ``expert load-balancing'') runs a primal greedy strategy and for the dual \textbf{any} prediction-with-experts algorithm, restarting at timestep $n/2$. 

\new{
	\begin{algorithm}[H]
		\caption{\expertLB}
		\begin{algorithmic}[0]                 
                  \For{each time $t$}
                  \State let $\A^t$ be the arriving random item 
                  \State \textbf{primal step}: choose $\p^t \in \Delta^k$ to minimize $\ip{\w^t \A^t}{\p^t}$ \Comment{\emph{play best-response}}
                  \State \textbf{dual step}: feed $\A^t \p^t$ to a prediction-with-experts algorithm (whose current state is $\w^t$)
                  \State \phantom{~~~~~~~~~~~~~~~~} let $\w^{t+1}$ be the resulting state of the algorithm. 
                  \If{$t = n/2$} 
                         \State restart the state of the prediction-with-experts algorithm \Comment{\emph{restart algorithm at half-time}}
                  \EndIf
                  \EndFor
                \end{algorithmic}
              \end{algorithm}
}             



Note the simplicity of the algorithm: it is perhaps the ``natural'' algorithm, once we decide to reduce load-balancing to a prediction-with-experts algorithm. 
 \new{The only non-intuitive step is perhaps that we restart the process at time $n/2$. This ensures that at each step, the current state depends on at most $n/2-1$ of the random choices from the $A^1, \ldots, A^n$, and hence the next random item still has ``enough randomness'' for our analysis to go through\footnote{The formal reason why we need this restart is that Lemma \ref{lemma:maximalBernstein} (because of Theorem \ref{thm:pruss}) can only handle up to half of the random sequence.}. Also, another interpretation of the algorithm is that the dual vectors $\w^t \in \Delta^m$ are providing (adaptive) aggregations of the $m$ machines into a single one, in which case the best action is just to minimize the load on this machine (which is exactly what our best-response primal step is doing).}
 
 The following theorem shows if the offline optimal load $\lambda^*$ is not too small, then \expertLB obtains value almost equal to $\lambda^*$ up to the regret of the prediction-with-experts algorithm used.
 


\full{
		\begin{thm}[Load-balancing guarantee] \label{thm:loadBalHP}
			Suppose $\{A^t\}_t$ is $(M,\gamma)$-well-bounded load-balancing instance for $\gamma \geq 1$. Let $\lambda^*$ be its optimal load and suppose $\e \leq \e_1 := \frac{1}{84}$ and $\delta \in (0, \e]$ are such that $\lambda^* \ge \frac{3M \log(m/\delta)}{\e^2}$. \new{If the algorithm \expertLB is run with \textbf{any} prediction-with-experts algorithm with $(\regA, R)$-regret}, then it returns a solution $\{\p^t\}_t$ such that with probability at least $1 - \delta$ $$\bigg\|\sum_t \A^t \p^t - \regA \sum_t \left|\A^t \p^t\right|\bigg\|_{\max} \le \lambda^* (1 + c_1\gamma \e) + 2R,$$ for a universal constant $c_1$. In particular, if the \MW algorithm from Theorem \ref{thm:MW} is used with learning parameter $\e$, the guarantee becomes (for a slightly larger constant $c_1$)
			$$\bigg\|\sum_t \A^t \p^t - \e \sum_t \left|\A^t \p^t\right|\bigg\|_{\max} \le \lambda^* (1 + c_1\gamma \e).$$
		\end{thm}
               
	The case when the load matrices are non-negative is of independent interest (generalizing scheduling on unrelated machines); in this case well-boundedness means $A^t_{ij} \in [0,M]$ and we can set $\gamma = 1$ to get:
                \begin{cor}[Positive loads guarantee]
                  \label{cor:loadBalPos}
                  Suppose $\{A^t\}_t$ is a load-balancing instance with $A^t_{ij} \in [0,M]$, and optimal load $\lambda^*$. Suppose $\e \leq \e_1 := \frac{1}{84}$ and $\delta \in (0, \e]$ satisfy $\lambda^* \ge \frac{3 M \log(m/\delta)}{\e^2}$. Given $n$, $M$ and $\e$, the algorithm \expertLB (with the \MW algorithm) finds an online solution $\{\p^t\}_t$ such that $\| \sum_t \A^t \p^t \|_\infty \leq (1+ O(\e))\lambda^*$ with probability at least $1 - \delta$.
                \end{cor}
}

\short{
	\begin{thm}[Load Balancing Guarantee] \label{thm:loadBalHP}
			Suppose $\{A^t\}$ is $(M,\gamma)$-well-bounded load-balancing instance for $\gamma \geq 1$. Let $\lambda^*$ be its optimal load and suppose $\e \leq \e_1 := \frac{1}{84}$ and $\delta \in (0, \e]$ are such that $\lambda^* \ge \frac{3M \log(m/\delta)}{\e^2}$. Given values of $n$, $M$ and $\e$, the algorithm \expertLB finds an online solution $\{\p^t\}_t$ such that with probability at least $1 - \delta$, $\left\|\sum_t \A^t \p^t - \e \sum_t |\A^t \p^t|\right\|_{\max} \le \lambda^* (1 + c_1(1+\gamma) \e)$ for a universal constant $c_1$.
			
			Moreover, if all $A^t$'s are non-negative, we can take $\gamma = 1$ and then have $\| \sum_t \A^t \p^t \|_\infty \leq (1+ O(\e))\lambda^*$ with probability at least $1 - \delta$.
	\end{thm}
}

\full{
In the rest of \S\ref{sec:lb} we prove Theorem~\ref{thm:loadBalHP}. We first show in \S\ref{sec:guar-expect} that the maximum load is close in expectation to $\lambda^*$. We then extend these ideas to prove the high-probability guarantee in \S\ref{sec:guar-whp}.
}


\full{	\subsection{The guarantee in expectation}}
\short{\subsection{Analysis of \expertLB}}

        \label{sec:guar-expect}		

\short{In this section we outline the analysis of algorithm \expertLB.}
Let $\optp^{1}, \ldots, \optp^{n}$ be the optimal solution for the offline instance. We see this as a mapping from matrix $A^t$ to solution $\optp^{t}$ (say $\phi(A^t) = \optp^t$); this way, define $\boptp^t$ as the random solution with respect to the random matrix $\A^t$ (i.e. $\boptp^t = \phi(\A^t)$). 
	
	To simplify the notation, let $\o^t := \A^t \p^t$ denote the load vector incurred at step $t$ by our algorithm. Also, for an integer $\ell$, we use $\A^{\le \ell}$ to denote the sequence $\A^1, \ldots, \A^{\ell}$, and similarly for other objects. 
	
	 Theorem~\ref{thm:loadBalHP} seeks to bound $\left\|\sum_t \o^t - \regA \sum_t |\o^t|\right\|_{\max}$. 
%
%
%
%
	By the exchangeability of sampling without replacement, $(\A^1, \ldots, \A^{n/2})$ has the same distribution as $(\A^{n/2+1}, \ldots, \A^n)$. This implies that our algorithm, with its half-time reset, behaves in the same way in the two halves of the process, namely $(\A^{\le n/2}, \p^{\le n/2})$ has the same distribution as $(\A^{> n/2}, \p^{> n/2})$; thus, it suffices to analyze the first half.
	
	\begin{fact}[Suffices to analyze first half] \label{fact:loadBal2}
	The random variables  $\big\|\sum_{t\le n/2} \o^t - \regA \sum_{t \le n/2} |\o^t| ~\big\|_{\max}$ and $\big\|\sum_{t > n/2} \o^t - \regA \sum_{t > n/2} |\o^t| ~\big\|_{\max}$ have the same distribution. 
	\end{fact}
		
	Now we need that the computed dual solutions $\w^t$ capture our (non-linear) total load, i.e., that $\sum_t \ip{\w^t}{\o^t} \approx \|\sum_t \o^t \|_{\max}$; but this follows directly from the fact we used an $(\regA, R)$-regret prediction-with-experts algorithm.
	
	\begin{fact}[Dual captures load] \label{fact:loadBal3}
		For every scenario we have
		\begin{align*}
			\sum_{t = 1}^{n/2} \ip{\w^t}{\o^t} ~~\ge~~ \bigg\|\sum_{t \le n/2} \o^t - \regA \sum_{t \le n/2} |\o^t|\,\bigg\|_{\max} - R.
		\end{align*}
	\end{fact}
	
	Furthermore, let $\bopto^t := \A^t \boptp^t$ denote the load incurred by the optimal solution $\boptp$ at step $t$. By our primal greedy choice of the $\p^t$'s, we directly have the following.
	
	\begin{fact}[Optimality of algorithm wrt duals] \label{fact:loadBal4}
	\begin{align*}
		\sum_{t = 1}^{n/2} \ip{\w^t}{\o^t} = \sum_{t = 1}^{n/2} \w^t \A^t \p^t \le \sum_{t = 1}^{n/2} \w^t \A^t \boptp^t = \sum_{t = 1}^{n/2} \ip{\w^t}{\bopto^t}.
	\end{align*}
	\end{fact}

	Together these give $\big\|\sum_{t\le n/2} \o^t - \regA \sum_{t \le n/2} |\o^t| ~\big\|_{\max} \le R + \sum_{t = 1}^{n/2} \ip{\w^t}{\bopto^t}$ and so in order to upper bound our load in expectation it suffices to show that
		\begin{align}
			\E\bigg[\sum_{t \le n/2} \ip{\w^t}{\bopto^t}\bigg] \lesssim \sum_{t \le n/2} \ip{\E[\w^t]}{\E[\bopto^t]} \le \bigg\|\E\bigg[\sum_{t \le n/2} \bopto^t\bigg]\bigg\|_{\max} = \frac{\lambda^*}{2}. \label{eq:guaranteeLB}
		\end{align}
	Notice that these are easy implications if we were sampling \emph{with replacement}, since in that case $\w^t$ and $\ho^t$ are independent. For the random permutation model we need to work harder because they are dependent: \new{$\w^t$ is determined by the history $\A^{<t}$ and the optimal solution's load $\ho^t$ at time $t$ depends on the history $\A^{<t}$ (since items that have occurred in the history cannot occur at time $t$).}

	However, this dependence of $\ho^t$ on the history $\A^{<t}$ should be quite small. In fact this is the main tool for our analysis: we show that $\E[\o^{*t}_i \smid \A^{<t}] \approx \E\, \o^{*t}_i$ happens for most histories, for all $t, i$. More precisely, we show that for each $i$, with high probability $\E[\o^{*t}_i \smid \A^{<t}] = \E\,\o^{*t}_i \pm O(\e) \cdot (|\E\,\o^{*t}_i| + \frac{\gamma \lambda^*}{n})$ \emph{for all $t \le n/2$ simultaneously}. Note that applying Bernstein's inequality to each $\o^{*t}_i$ and taking a union bound over the $t$'s would only give $\E[\o^{*t}_i \smid \A^{<t}] = \E\,\o^{*t}_i \pm O(\e \log n) \cdot (|\E\,\o^{*t}_i| + \frac{\gamma \lambda^*}{n})$, with an extra $\log n$ factor. To avoid this we employ a \emph{maximal} version of Bernstein's inequality. To simplify the notation, we use the shorthand $\Et_t \mathbf{X} := \E[\mathbf{X} \smid \A^{<t}]$ for expectations conditioned on the history before time $t$. 

	\begin{lemma} \label{lemma:balancedEveryTime}
		Consider $i \in [m]$. Then with probability at least $1 - \frac{\delta}{m^2}$ we have \new{$\Et_t[\o^{*t}_i] \in \E\,\o^{*t}_i \pm 80 \e (|\E\,\o^{*t}_i| + \frac{\gamma \lambda^*}{n})$} for all $t \le \frac{n}{2}$ simultaneously. 
	\end{lemma}
	
	\begin{proof}
	Let $\mu = \frac{1}{n} \sum_{t \le n} A^t_i \optp^t \leq \frac{\lambda^*}{n}$ be the expected value of $\o^{*t}_i$, which is independent of $t$. Notice that in every scenario $\sum_{j = 1}^n \o^{*j}_i - \sum_{j < t} \o^{*j}_i$ gives the total $i^{th}$ load from items that have not shown up at times $j < t$. Therefore, the conditional expected $i^{th}$ load at time $t$ is the average over these remaining items:
	\begin{equation}
	\E_t[\o^{*t}_i] = \frac{\sum_{j = 1}^n \o^{*j}_i - \sum_{j < t} \o^{*j}_i}{n - (t-1)} = \frac{n \mu - \sum_{j < t} \o^{*j}_i}{n - (t-1)}. \label{eq:expOhat}
	\end{equation}
 Thus, to control $\E_t[\o^{*t}_i]$ it suffices to control the sum $\sum_{j < t} \o^{*j}_i$, which should be highly concentrated around the mean $(t-1) \mu$.  
 	
 	For that, let $\alpha = n |\mu| + \gamma\lambda^*$; we want to use the maximal Bernstein's inequality for sampling without replacement (Lemma \ref{lemma:maximalBernstein}) to bound the maximum deviation $\max_{t \le n/2} |\sum_{j\le t} \o^{*j}_i - t\mu| \le 40 \e \alpha$ with probability at least $1-\frac{\delta}{m^2}$. 
 	 For that, we need to bound the variance $V := \frac{1}{n} \sum_{t \le n} (A^t_i p^{*t} - \mu)^2 = \E(\o^{*t}_i - \E\,\o^{*t}_i)^2$. Since $\o^{*t}_i \in [-M, M]$, we have that $V \le M \cdot \E\left[\left|\o^{*t}_i - \E\,\o^{*t}_i\right|\right]$. Then employing Lemma \ref{lemma:wellBounded} to bound this absolute value we get $V \le M (2|\mu| + \frac{2\gamma\lambda^*}{n}) = 2M\alpha/n$. Thus, applying the maximal Bernstein's inequality Lemma \ref{lemma:maximalBernstein} we obtain 
 	\begin{gather*}
 	 \Pr\left(\max_{t \le n/2} |{\textstyle \sum_{j\le t}} \o^{*j}_i - t\mu| \le 40 \e \alpha \right) \leq 30 \exp\left(-\left(\frac{40 \e}{24}\right)^2 \frac{\alpha^2}{n V + 40 \e \alpha (M/24)}\right) \\
 	 \le 30 \exp\left(- \left(\frac{40 \e}{24}\right)^2 \frac{\alpha}{(2 + 40 \e/24) M}\right) \leq 30 \exp\left(- \frac{\e^2 \gamma\lambda^*}{M}\right) \le \frac{\delta}{m^2},
 	\end{gather*} 
 	where we used $\delta \le \e \leq \frac{1}{80}$, $\gamma \geq 1$, and $\lambda^* \geq \frac{3 M \log (m/\delta)}{\e^2}$.

	Whenever this event holds equation \eqref{eq:expOhat} gives that for any $t \le n/2$
	\begin{equation*}
	\E_t[\o^{*t}_i] \in \frac{(n - (t-1)) \,\mu \pm 40\e \alpha}{n - (t-1)} = \mu \pm \frac{40\e \alpha}{n - (t-1)} \in \mu \pm \frac{80 \e\alpha}{n},
	\end{equation*}	 
	where the last inequality uses $t \le n/2$. The lemma then follows from the definition of $\alpha$.
	\end{proof}

\full{
	Using this lemma it is not difficult to show that inequality \eqref{eq:guaranteeLB} approximately holds. We sketch how this can be done in order to illustrate that little of the structure of the problem is actually used (and thus might carry over to other settings). This inequality plus Facts \ref{fact:loadBal2}--\ref{fact:loadBal4} imply that the bound in Theorem \ref{thm:loadBalHP}  holds in expectation. \new{(We remark that this guarantee in expectation does not require Lemma \ref{lemma:balancedEveryTime} to hold simultaneously for all $t \le n/2$; this will only be required below for our guarantee with high probability.)}
	
	\begin{lemma} \label{lemma:optLoadBalExp}
		We have $\E\left[\sum_{t \le n/2} \ip{\w^t}{\ho^t}\right] \le \frac{\lambda^*}{2} + O(\e \gamma  \lambda^*)$.
	\end{lemma}
}
\short{
	Using this lemma, it is an easy exercise to show that \eqref{eq:guaranteeLB} indeed holds (within a factor $(1 \pm \gamma \e)$); this then gives that the bound in Theorem \ref{thm:loadBalHP}  holds in expectation.
	
	Moreover, we can show that \emph{with probability at least $1- \delta$}, $\sum_{t \le n/2} \ip{\w^t}{\ho^t} \le \frac{\lambda^*}{2} + O(\e \gamma) \lambda^*$. For that we need to employ martingale concentration (Freedman's inequality), and use Lemma \ref{lemma:balancedEveryTime} to control the predictable variation of our martingale. This then gives the full statement of Theorem \ref{thm:loadBalHP}.
}
		
	\begin{proof}[Proof sketch]
%

    Let $\cE$ denote the event that the approximation from Lemma \ref{lemma:balancedEveryTime} holds for every $i \in [m]$, which from a union bound then holds with probability at least $1 - \frac{\delta}{m}$. Let us split the expectation in the left-hand side of the lemma based on $\cE$:
		\begin{align}
			\E\bigg[\sum_{t \le n/2} \ip{\w^t}{\bopto^t}\bigg] = \E\bigg[\sum_{t \le n/2} \ip{\w^t}{\bopto^t} \;\bigg\vert\; \cE\bigg] \Pr(\cE) + \E\bigg[\sum_{t \le n/2} \ip{\w^t}{\bopto^t} \;\bigg\vert\; \overline{\cE}\bigg] \Pr(\overline{\cE}). \label{eq:lbExpParts}
		\end{align}
		To upper bound the first term we further split the expectation by conditioning on the history up to time $t-1$ and recalling that $\w^t$ is defined by it:
		\begin{align}
			\E\bigg[\sum_{t \le n/2} \ip{\w^t}{\bopto^t} \;\bigg\vert\; \cE\bigg] = \sum_{t \le n/2} \E\Big[ \E\Big[\ip{\w^t}{\bopto^t} \smid \A^{<t}, \cE \Big] \;\Big\vert\; \cE\Big] = \sum_{t \le n/2} \E\Big[ \ip{\w^t}{\E[\o^{*t} \smid \A^{<t}, \cE]} \;\Big \vert\; \cE\Big]. \label{eq:sketch1}
		\end{align}
		Recalling the guarantee of Lemma \ref{lemma:balancedEveryTime} for $\cE$, we can bounds the innermost expectation by
		\begin{align*}
			\E[\o^{*t} \smid \A^{<t}, \cE] \le \E\,\o^{*t} + 80 \e |\E\,\o^{*t}| + O(\e \gamma) \frac{\lambda^*}{n} \ones \le \max\{(1+80\e) \E\,\o^{*t}, 0\} + O(\e \gamma) \frac{\lambda^*}{n} \ones \le (1 + O(\e \gamma)) \frac{\lambda^*}{n} \ones,
		\end{align*}
		where the second inequality uses that $\e \le \frac{1}{80}$ and the last one uses the optimality $\o^{*t} \le (\lambda^*/n) \ones$. Plugging this bound onto inequality \eqref{eq:sketch1} gives that we can upper bound the first term in equation \eqref{eq:lbExpParts} by $(1 + O(\e \gamma)) \frac{\lambda^*}{2}$.
		
		To upper bound the second term in equation \eqref{eq:lbExpParts}, using the well-boundedness of the instance it is not hard to show that in every scenario the sum of optimal loads satisfies $\sum_{t \le n/2} \o^{*t}_i \le O(\gamma) \lambda^*$, and thus $\sum_{t \le n/2} \ip{\w^t}{\o^{*t}} \le O(\gamma)\, m \lambda^*$. Since $\bar{\cE}$ only holds with probability $\frac{\delta}{m} \le \frac{\e}{m}$, we can upper bound the second term in equation \eqref{eq:lbExpParts} by $O(\e \gamma)\,\lambda^*$. 
		
		Putting these bounds together concludes the proof of the lemma. 
	\end{proof}
		

\ifproc
\else
	
	\subsection{The guarantee with high probability}
        \label{sec:guar-whp}
	We now show that Lemma~\ref{lemma:optLoadBalExp} actually holds with high probability. 
	
	\begin{lemma} \label{lemma:concentration}
		With probability at least $1 - 3 \delta$ we have
		\begin{align*}
			\tsty \sum_{t \le n/2} \ip{\w^t}{\ho^t} \le (1 + 168\e \gamma) \frac{\lambda^*}{2}.
		\end{align*}
	\end{lemma}
	Again, if we were sampling \emph{with replacement}, then it would not be difficult to show that $\left(\ip{\w^t}{\ho^t} - \ip{\w^t}{\E\,\ho^t}\right)_{t}$ is a martingale difference sequence, and consequently martingale concentration inequalities would imply $\sum_{t \le n/2} \ip{\w^t}{\ho^t} \lesssim \sum_{t \le n/2} \ip{\w^t}{\E\,\ho^t} \le \sum_{t \le n/2} \|\E\,\ho^t\|_{\max} \lesssim \frac{\lambda^*}{2}$ with good probability. For our case of sampling without replacement, we  proceed similarly, but using the martingale difference sequence $\left(\ip{\w^t}{\ho^t} - \ip{\w^t}{\Et_t[\ho^t]}\right)_{t}$ (notice the conditioning on the history), and then using Lemma \ref{lemma:balancedEveryTime} to get that with good probability for all $t$ we have $\Et_t[\ho^t] \approx \E[\ho^t]$.
	
	So define the random variable $\bY_t := \ip{\w^t}{\ho^t - \E_t[\ho^t]}$. Since $\w^t$ is determined by $\A^{<t}$, we have that $\Et_t[\bY_t] = 0$ and hence $(\bY_t)_t$ is indeed a martingale difference sequence (adapted to $(\A^t)_t$). In order to obtain concentration for this martingale, we first need to control its conditional variance, which is accomplished in the next lemma. 
	
	\begin{lemma} \label{lemma:predVariation}
			With probability at least $1 - \frac{\delta}{m}$ we have 
      \new{\[ \sum_{t \le n/2} \Var(\bY_t \smid \A^{<t}) \le 4M \cdot \sum_{t \le n/2} \ip{\w^t}{\left|\E\,\ho^t\right|} + 4 M \gamma\lambda^*. \]}
	\end{lemma}
	
	\begin{proof}
		Since $\o^{*t}_i \in [-M,M]$, we have $\bY_t \in [-2M, 2M]$ and hence $\bY^2_t \le 2M |\bY_t|$. Since $\Et_t[\bY_t] = 0$, for all the histories we have
		\begin{align*}
			\Var(\bY_t \smid \A^{<t}) &= \Et_t[\bY_t^2] \le 2M \cdot \Et_t |\bY_t| = 2M \cdot \Et_t\bigg[\Big\vert\sum_i \w^t_i \left(\o^{*t}_i - \Et_t[\o^{*t}_i] \right)\Big\vert \bigg]\\
			&\le 2M \, \sum_i \w^t_i \, \Et_t\left[\left|\o^{*t}_i - \Et_t[\o^{*t}_i]\right| \right] \stackrel{Lemma~\ref{lemma:wellBounded}}{\le} 4M \, \ip{\w^t}{\left|\Et_t[\o^{*t}]\right|} + \frac{4 M \gamma \lambda^*}{n},
		\end{align*}
	where the last inequality uses $\w^t \in \Delta^m$.
	
	To conclude the proof, we need the approximation $\E[\o^{*t}] \approx \Et_t[\o^{*t}]$. For that, let $\cE$ be the event that the guarantees from Lemma \ref{lemma:balancedEveryTime} hold for all $i \in [m]$; by a union bound, we have that $\cE$ happens with probability at least $1 - \frac{\delta}{m}$. Under the event $\cE$, we have that for all $t \le n/2$
	\begin{align*}
		|\Et_t\, \o^{*t}_i| \le |\E\,\o^{*t}_i| + 80 \e |\E\,\o^{*t}_i| + \frac{80 \e \gamma \lambda^*}{n} \le 2|\E\,\o^{*t}_i| + \frac{\gamma \lambda^*}{n},
	\end{align*}
	where the last inequality uses $\e \le \frac{1}{80}$. Thus, under the event $\cE$ we have $\Var(\bY_t \smid \A^{<t}) \le 8M \ip{\w^t}{|\E\,\o^{*t}|} + \frac{8 M \gamma \lambda^*}{n}$. Adding over all $t \le n/2$ then gives the lemma.
	\end{proof}
	
	Now we can use our martingale to prove Lemma \ref{lemma:concentration}.
	
	\begin{proof}[Proof of Lemma \ref{lemma:concentration}]
		We first employ Freedman's inequality to the martingale $(\bY_t)_t$. For that, let $\alpha := 4 \sum_{t \le n/2} \ip{\w^t}{|\E\, \o^{*t}|} + 4 \gamma \lambda^*$, so that $M\alpha$ upper bounds the conditional variance $\sigma^2 := \sum_{t \le n/2} \Var(\bY_t \smid \A^{<t})$ with probability $1-\frac{\delta}{m}$. Then from Freedman's inequality (Theorem \ref{thm:freedman}) we get that
		\begin{align*}
			\Pr\bigg(\sum_{t \le n/2} \bY_t > \e \alpha \textrm{ and } \sigma^2 \le M\alpha \bigg) \le \exp\left(- \frac{\e^2 \alpha^2}{4 M \alpha} \right) \le \exp\left(- \frac{\e^2 \gamma\lambda^*}{M} \right) \le \frac{\delta}{m},
		\end{align*}
		where the last inequality uses $\gamma \geq 1$ and $\lambda^* \geq \frac{3 M \log (m/\delta)}{\e^2}$. Since $\Pr(\sigma^2 \le M\alpha) \ge 1- \frac{\delta}{m}$, a union bound gives that with probability at least $1-\frac{2\delta}{m}$ we have $\sum_{t \le n/2} \bY_t \le \e \alpha$. Expanding the definition of $\bY$ and $\alpha$, this is equivalent to
		\begin{align}
			\sum_{t \le n/2} \ip{\w^t}{\o^{*t}} \le \sum_{t \le n/2} \ip{\w^t}{\Et_t\, \o^{*t} + 4 \e |\E\,\o^{*t}|} + 4 \e \gamma \lambda^*. \label{eq:martingale1}
		\end{align}
		To upper bound the first term in the right-hand side, employing Lemma \ref{lemma:balancedEveryTime} to each $i \in [m]$ we have that with probability at least $1-\frac{\delta}{m}$ the following holds for all $t \le n/2$: 
		\begin{align*}
			\Et_t\,\o^{*t} + 4\e|\E\,\o^{*t}| \le \E\,\o^{*t} + 84\e |\E\,\o^{*t}| + \frac{80 \e \gamma \lambda^*}{n} \ones \le \max\{(1+84\e) \E\,\o^{*t}, 0\} + \frac{80\e \gamma \lambda^*}{n} \ones \le (1 + 164 \e \gamma)\frac{\lambda^*}{n},
		\end{align*}	
		where the second inequality follows from the fact $\e \le \frac{1}{84}$ and $\E\,\o^{*t} \le (\lambda^*/n) \ones$. Plugging this bound on inequality \eqref{eq:martingale1} and taking a union bound, we have that with probability at least $1 - \frac{3\delta}{m}$
		\begin{align*}
			\sum_{t \le n/2} \ip{\w^t}{\o^{*t}} \le (1 + 168 \e \gamma) \frac{\lambda^*}{2}.
		\end{align*}	
		This concludes the proof. 
	\end{proof}

	\begin{proof}[Proof of Theorem \ref{thm:loadBalHP}]
		From triangle inequality we have
		\begin{align*}
			\bigg\|\sum_t \o^t - \regA \sum_t \left|\o^t\right|\bigg\|_{\max} \le \bigg\|\sum_{t \le n/2} \o^t - \regA \sum_{t \le n/2} \left|\o^t\right|\bigg\|_{\max} +\bigg\|\sum_{t > n/2} \o^t - \regA \sum_{t > n/2} \left|\o^t\right|\bigg\|_{\max}. 
		\end{align*}
		Using Facts \ref{fact:loadBal3} and \ref{fact:loadBal4} and Lemma \ref{lemma:concentration}, with probability at least $1 - 3 \delta$ the first term in the right-hand side is at most $(1 + 168\e\gamma) \frac{\lambda^*}{2} + R$. Since from Fact \ref{fact:loadBal2} the same bound holds for the second term of the right-hand side, a union bound concludes the proof of the theorem with constant $c_1 = 168$.
	\end{proof}

\fi



\newcommand{\estopt}{\widehat{\OPT}}
\newcommand{\erropt}{\widetilde{o}}

	\section{Solving packing-covering multi-choice LPs given estimate of OPT}
        \label{sec:lp-to-lb}

This section shows how to solve packing/covering multiple-choice (\PCMC) LPs  via a simple reduction to the generalized load-balancing problem from the previous section, provided an estimate of the optimal value is available. In the following sections we show how to get such an estimate.


From \S\ref{sec:introModel}, recall that a \PCMC LP has the form: 
	\begin{align}
		\max ~\tsty \sum_{t = 1}^n \pi^t x^t & \tag{\PCMC} \label{eq:pcmc2}\\
		  st ~\tsty \sum_{t = 1}^n A^t x^t &\le b \label{eq:pack}\\
		     \tsty \sum_{t = 1}^n C^t x^t &\ge d \label{eq:cover}\\
		     x^t &\in \fullsimplex^k \ \ \ \forall t \in [n], \label{eq:mc}
	\end{align}
	where all the data $\pi^t, A^t, C^t, b, d$ is \textbf{non-negative} and $\fullsimplex^k$ denotes the ``full simplex'' $\{x \in [0,1]^k : \sum_j x_j \le 1\}$. We use $m_p$ to denote the number packing constraints \eqref{eq:pack} and $m_c$ to denote the number of covering constraints \eqref{eq:cover}, so our matrices have dimensions $A^t \in \R_+^{m_p \times k}$, and $C^t \in \R_+^{m_c \times k}$. 
	Given a \PCMC LP $\L$, we use $\OPT(\L)$ to denote its optimal value; we use simply $\OPT$ when the LP is clear from context.

	In the \emph{online} version of the problem there is a fixed \PCMC LP $\L$ whose blocks $\{(\pi^t, A^t, C^t)\}_t$ are presented one-by-one in random order. More precisely, the number of time steps $n$ and the right-hand sides $b, d$ are known upfront, and 
        the triples $(\bpi^1, \A^1, \C^1), \ldots, (\bpi^n, \A^n, \C^n)$ are sampled from $\{(\pi^t, A^t, C^t)\}_t$ uniformly \emph{without replacement}. Then define the randomly permuted LP $\bL$
	\begin{align}
		\bL = \max\left\{\sum_t \bpi^t x^t : \sum_t \A^t x^t \le b, ~\sum_t \C^t x^t \ge d, ~x^t \in \fullsimplex^k ~~\forall t\right\}
	\end{align}
	 At time step $t$, the algorithm computes a (random) vector
         $\bx^t \in \fullsimplex^k$ based on the information  seen up to
         time $t$, i.e., $(\bpi^1, \A^1, \C^1), \ldots, (\bpi^t, \A^t, \C^t)$, plus $n$ and $b,d$. We call the sequence $\{\bx^t\}_t$ an \emph{online solution}. 

We say that an online solution $\{\bx^t\}_t$ is \emph{$\e$-feasible} for $\bL$ if it satisfies the packing constraints (i.e., $\sum_{t \le n} \A^t \bx^t \le b$) and \emph{almost} satisfies the covering constraints (i.e., $\sum_{t \le n} \C^t \bx^t \ge (1 - \e) d$). The goal in the online \PCMC LP problem is to obtain an online solution $\{\bx^t\}_t$ which (with high probability) is $\e$-feasible, and obtains a reward $\sum_{t \le n} \bpi^t \bx^t \geq (1-O(\e))\OPT(\bL)$. Notice that $\OPT(\bL) = \OPT(\L)$, so again we are comparing against the optimal offline solution.

\full{
%
%
%

}

	\subsection{The algorithm \LPtoLB} \label{sec:LPtoLB}

We now give an algorithm \LPtoLB that solves \PCMC LPs via a natural reduction to the \loadBal problem.
\new{The idea is simple: The \loadBal problem can be thought as a congestion-minimization problem $\min\{\lambda : \sum_{t = 1}^n A^t p^t \le \lambda \cdot \ones, ~p^t \in \Delta^k ~\forall t\}$, which is reminiscent of a \emph{packing-only} LP where all right-hand sides are the same. To bring a \PCMC LP into this form, we introduce the objective function as a covering constraint $\sum_{t=1}^n \pi^t x^t \ge \estopt$ (thereby requiring an estimate of \OPT), multiply the covering constraints by $-1$ to obtain packing constraints, and shift/scale to ensure uniform right-hand sides.} Observe that the reduction creates negative entries on the left-hand side, which is why we considered the \loadBal problem with possibly negative loads.
 

	More precisely, given a \PCMC LP $\L$ and an estimate $\estopt$ of $\OPT$, define the matrices $H^1, \ldots, H^n$ with $k + 1$ columns and $m_p+m_c+1$ rows (indexed from $0$ to $m := m_p+m_c$) as follows:
the zeroth row of $H^t$ equals the vector 
\[ H^t_{0,
  \star} := \textstyle \big(\frac{2}{n} - \frac{\pi^t_1}{\estopt}, \ldots, \frac{2}{n} - \frac{\pi^t_k}{\estopt}, \frac{2}{n}\big)~; \] for $i \in \{1, \ldots, m_p\}$, the $i^{th}$ row of $H^t$ is 
\[ H^t_{i,
  \star} := \textstyle \big(\frac{a^t_{i1}}{b_i}, \ldots,
\frac{a^t_{ik}}{b_i}, 0\big), \] and for $i \in \{m_p+1,
\ldots, m_p+m_c\}$, the $(m_p+i)^{th}$ row of $H^t$ is \[ H^t_{m_p+i,
  \star} := \textstyle \big(\frac2n - \frac{c^t_{i1}}{d_i}, \ldots, \frac2n - \frac{c^t_{ik}}{d_i}, \frac{2}{n} \big).\]		

The algorithm \LPtoLB can be thought of as running in phases, where it constructs the instance $\{H^t\}_t$, feeds it to the algorithm \expertLB, and outputs a scaling of the solution returned. 

\new{
	\begin{algorithm}[H]
		\caption{. \LPtoLB}
		\begin{algorithmic}[0]
			\State \textbf{Input:} Reals $\estopt, \e, \delta > 0$, number $n$ of items, and right-hand sides $b$ and $d$.
			\medskip
			\State \textbf{Phase 1:} Compute matrices $H^1, \ldots, H^n$ relative to the \PCMC LP $\L$.
			\medskip
			\State \textbf{Phase 2:} Run algorithm \expertLB (with the \MW algorithm from Theorem \ref{thm:MW}) over the generalized load-balancing instance $\{H^t\}_t$ with parameters $M = \frac{2\e^2}{\log ((m + 1) / \delta)}$, $\e' = \sqrt{8} \e$ and $\delta$ (given), and let $\{\tilde{\bx}^t\}_t$ be the returned solution.
			\medskip
			\State \textbf{Phase 3:} Define $\widehat{\bx}^t_j := \frac{(1 - \e')}{(1 + 4 c_1 \e')}\tilde{\bx}^t_j$ for $j = 1, \ldots, k$. Return $\{\widehat{\bx}^t\}_t$ as the solution for the \PCMC LP $\L$.
		\end{algorithmic}
	\end{algorithm}
}
	
	Some remarks are in order. Firstly, for simplicity we described the algorithm as running in phases, but it is easy to implement this in an online fashion. Indeed, when at time $t$ the algorithm sees a random triple $(\bpi^t, \A^t, \C^t)$, it can construct the appropriate matrix $\mathbf{H}^t$, feed it to algorithm \expertLB to get back the vector $\tilde{\bx}^t$, and output the scaled version $\widehat{\bx}^t$. \new{Secondly, we used the algorithm \expertLB instantiated with the \MW algorithm, but this is for ease of exposition. We use the guarantees from Theorem~\ref{thm:loadBalHP} in a black-box fashion, and hence could use any algorithm for the prediction-with-experts problem (and obtaining different guarantees).}

        The guarantee of algorithm \LPtoLB depends on the \emph{generalized width} of the \PCMC LP, defined as:
        \begin{gather}
          \min\left\{\min_{i,j,t} \frac{b_i}{a^t_{i,j}}, ~~ \min_{i,j,t} \frac{d_i}{c^t_{i,j}}, ~~ \min_{t,j} \frac{\OPT}{\pi^t_j} \right\}. \label{eq:gw}
        \end{gather}

The next theorem shows that algorithm $\LPtoLB$ with high probability returns an $\e$-feasible $(1-\e)$-optimal solution as long as the generalized width is $\Omega(\frac{\log(m/\delta)}{\e^2})$ and a $(1-\e)$-optimal under estimate $\estopt$ is available.

. 

		\begin{thm} \label{thm:LPAssEstimates}
Consider a feasible packing/covering multiple-choice LP $\L$. Consider $\e > 0$ at most a sufficiently small constant $\e_0$, and $\delta \in (0, \e]$, and assume that the generalized width of $\L$ is at least $\frac{\log ((1 + m_p + m_c)/\delta)}{\e^2}$. 
Then \LPtoLB finds an online solution $\{\bx^t\}$ that with probability at least $1 - \delta$ is $(c_2 \e)$-feasible for $\bL$ and has value $\sum_t \bpi^t \bx^t \geq (1 - c_3 \e) \estopt$, for constants $c_2, c_3 \geq 1$.
		\end{thm}



\short{The analysis of \LPtoLB relies on: 1) making sure the assumptions of Theorem~\ref{thm:loadBalHP} are satisfied by the load-balancing instance $\{H^t\}_t$, and 2) connecting feasible solutions for the instance $\{H^t\}_t$ to solutions of the LP $\L$. These give Theorem \ref{thm:LPAssEstimates}.
}

\ifproc
\else

	\subsection{Proof of Theorem~\ref{thm:LPAssEstimates}} 


        The idea is to first show that the reduction creates a load-balancing instance that is well-bounded, and whose optimal load is not very small. With these guarantees, we apply Theorem~\ref{thm:loadBalHP} from \S\ref{sec:lb} to get a good load-balancing solution. Interpreting this as a solution to our \PCMC LP will complete the proof.

	Assume that all hypotheses in the statement of Theorem~\ref{thm:LPAssEstimates} hold. Let $\{H^t\}_t$ be the generalized load-balancing instance relative to the input \PCMC LP $\L$, and let $\lambda^*$ denote the optimal value of the instance $\{H^t\}_t$. We start by bounding $\lambda^*$ by relating the solutions for the instance $\{H^t\}_t$ to the solutions for the LP $\L$. 
			
	\begin{claim} \label{claim:redLBOpt}
          $\frac{1}{2} \le \lambda^* \le 1$.
	\end{claim}
	
	\begin{proof}
          (Upper bound.) Let $\{x^t\}_t$ be an optimal solution for the \PCMC LP $\L$, and for each $t$ define $\tilde{x}^t = (x^t_1, \ldots, x^t_k, 0)$ by appending a zero at the end of $x^t$; it suffices to show that $\{\tilde{x}^t\}_t$ is a solution for the load-balancing instance $\{H^t\}_t$ with load at most 1.
          
          The feasibility of $\{x^t\}_t$ guarantees that for $i \in [m_p]$ we have $\sum_t H^t_{i,\star} \tilde{x}^t = \frac{1}{b_i} \sum_t A^t_{i,\star} x^t \le 1$, and for $i \in [m_c]$ we have $\sum_t H^t_{m_c+i,\star} \tilde{x}^t = 2 - \frac{1}{d_i} (\sum_t C^t_{i,\star} x^t) \le 1$. Moreover, the optimality of $\{x^t\}_t$ guarantees that $\sum_t H^t_{0,\star} \tilde{x}^t = 2 - \frac{1}{\estopt}(\sum_t \pi^t x^t) = 2 - \frac{\OPT}{\estopt} \le 1$, where the last inequality follows from $\OPT \ge \estopt$. Thus, $\{\tilde{x}^t\}_t$ achieves load at most 1 for $\{H^t\}_t$.

(Lower bound.) Consider an optimal solution $\tilde{x}^1, \ldots, \tilde{x}^n \in \simplex^k$ for the load-balancing instance $\{H^t\}_t$, so $\lambda^* := \|\sum_t H^tx^t \|_{\max}$. Define $x^t = (\tilde{x}^t_1, \ldots, \tilde{x}^t_k)$ by dropping the last component from $\tilde{x}^t$.

For $i \in [m_p]$ we have $\lambda^* \ge \sum_t H^t_{i,\star} \tilde{x}^t = \frac{1}{b_i} \sum_t A^t_{i,\star} x^t$, and hence $\sum_t A^t \frac{x^t}{\lambda^*} \le b$. Moreover, for $i \in [m_c]$ we have $\lambda^* \geq \sum_t H^t_{m_p+i,\star} \tilde{x}^t = 2 - \frac{1}{d_i} \sum_t C^t_{i,\star} x^t$, and hence $\sum_t C^t \frac{x^t}{\lambda^*} \geq \frac{2 - \lambda^*}{\lambda^*} d \ge d$, where the last inequality uses $\lambda^* \leq 1$ which was proved in the first part of the claim. Thus, $\{\frac{x^t}{\lambda^*}\}_t$ is a feasible solution for the LP $\L$.

In addition, we have $\lambda^* \ge \sum_t H^t_{0,\star} \tilde{x}^t = 2 - \frac{1}{\estopt} \sum_t \pi^t x^t$, and hence the solution $\{\frac{x^t}{\lambda^*}\}_t$ obtains value $\sum_t \pi^t \frac{x^t}{\lambda^*} \ge \frac{2 - \lambda^*}{\lambda^*} \estopt$. This must be at most the optimal value $\OPT$ for the LP $\L$, so $\OPT \geq \frac{2 - \lambda^*}{\lambda^*} \estopt$, or $\lambda^* \ge \frac{2 \,\estopt}{\estopt + \OPT} \ge \frac{1}{2}$, where the last inequality follows from $\frac{\OPT}{2} \leq \estopt \leq \OPT$.
This concludes the proof.
	\end{proof}

	\begin{claim} \label{claim:lpWellBounded}
		The instance $\{H^t\}_t$ is $(M,4)$-well-bounded for $M = \frac{2\e^2}{\log ((m + 1) / \delta)}$.
	\end{claim}

	\begin{proof}
		Using the lower bound $\lambda^* \ge \frac{1}{2}$ from the previous claim, it suffices to show that each row of $H^t$ contains entries that all lie in the interval $[-M, \frac{2}{n}]$, or all in the interval $[-\frac{2}{n}, M]$. 

	For the zeroth row, the generalized width lower bound implies $H^t_{0,j} \ge -\frac{\pi^t_j}{\estopt} \ge -\frac{M \OPT}{2\estopt} \ge -M$, where the last inequality uses $\frac{\OPT}{2} \leq \estopt$; since $H^t_{0,j} \le \frac{2}{n}$, it follows that $H^t_{0,j} \in [-M, \frac{2}{n}]$. For $i \in [m_p]$, we have $0 \le H^t_{i,j} \le \smash{\frac{a^t_{i,j}}{b_i}} \le M$. And for $i \in [m_c]$, we have $H^t_{m_p+i,j} \le \frac{2}{n}$ and $H^t_{m_p+i,j} \ge -\smash{\frac{ c^t_{i,j}}{d_i}} \ge -M$. This concludes the proof.
	\end{proof}

	\begin{proof}[Proof of Theorem \ref{thm:LPAssEstimates}]
	
	Again let $M = \frac{2\e^2}{\log ((m + 1) / \delta)}$ and recall that $\e' = \sqrt{8} \e$. Since Claims \ref{claim:lpWellBounded} and \ref{claim:redLBOpt} guarantee, respectively, that the instance $\{H^t\}_t$ is $(M,4)$-well-bounded and has optimal value $\lambda^* \ge \frac{3M \log ((m+1)/\delta)}{(\e')^2}$, we can applying Theorem \ref{thm:loadBalHP} to obtain that the solution $\{\tilde{\bx}^t\}_t$ returned by the algorithm \expertLB satisfies with probability at least $1 - \delta$
	\begin{align}
		\forall i = 0, \ldots, m, \ \ \ \sum_{t} \bH^t_{i,\star} \tilde{\bx}^t - \e' \sum_{t} |\bH^t_{i,\star} \tilde{\bx}^t| &\le \lambda^* (1 + 4 c_1 \e') \le (1 + 4 c_1 \e'), \label{eq:LPRedGuarantee}
	\end{align}
	where the last inequality follows again from Claim \ref{claim:redLBOpt}. 
        Let $\cE$ be the event that this bound holds. For $i \in [m_p]$, since $H^t_{i, \star} = \frac{A^t_{i,\star}}{b_i}$ is non-negative, this means that conditioned on $\cE$ we have $\sum_t \sum_{j \le k} \A^t_{ij} \tilde{\bx}^t_j \le b_i \frac{(1 + 4 c_1 \e')}{(1-\e')}$ and so the solution $\{\widehat{\bx}^t\}_t$ output by \LPtoLB satisfies the packing constraints. 

	For $i = 0$ we have $|\bH^t_{0,\star} \tilde{\bx}^t| \le \frac2n + \frac{1}{\estopt} \sum_{j \le k} \bpi^t_j \tilde{\bx}^t_j$, and thus the left-hand side of inequality \eqref{eq:LPRedGuarantee} is at least 
	\begin{align*}
		\left(2 - \frac{1}{\estopt} \sum_t \sum_{j \le k} \bpi^t_j \tilde{\bx}^t_j\right) - \e' \left(2 + \frac{1}{\estopt} \sum_t \sum_{j \le k} \bpi^t_j \tilde{\bx}^t_j\right) = 2 (1 - \e') -  \frac{(1 + \e')}{\estopt} \sum_t \sum_{j \le k} \bpi^t_j \tilde{\bx}^t_j.
	\end{align*}
	So conditioned on $\cE$, inequality \eqref{eq:LPRedGuarantee} gives $$\sum_t \sum_{j \le k} \bpi^t_j \tilde{\bx}^t_j \ge \estopt \frac{(1-\e') (1-2\e' - c_1 \e')}{(1+ 4c_1 \e') (1 + \e')} = (1- O(\e'))\, \estopt,$$ where the last inequality uses the fact that $\e$ (and hence $\e'$) is at most a sufficiently small constant. Therefore, conditioned on $\cE$, we have that the solution output by \LPtoLB satisfies $\sum_t \bpi^t \widehat{\bx}^t = \frac{(1- \e')}{(1 + 4c_1 \e')}\sum_t \sum_{j \le k} \bpi^t_j \tilde{\bx}^t_j \ge (1 - O(\e'))\, \estopt$.  
	
	An identical proof shows that conditioned on $\cE$ each covering constraint is $O(\e')$-approximately satisfied, i.e., $\sum_t \C^t \widehat{\bx}^t \geq (1 - O(\e'))\,d$. This concludes the proof of the theorem.
	\end{proof}


\fi


\full{
\section{Estimating the optimal value based on sampled LPs}
\label{sec:estPC}

In ths preceding section, we saw how to solve a \PCMC LP $\L$ given an
estimate of $\OPT(\L)$. 
Now we show how to obtain such estimates via ``sampled LPs''. For that, we require the LP to satisfy a \emph{stability} condition. Loosely speaking this asks that the covering constraints are not ``very tight'' --- if we reduce the covering requirements by a little, the optimal value should not increase by a lot. First, a definition to simplify notation:

	\begin{definition}
  	\label{def:epsFeas-lps}
  Given a \PCMC LP $\L$ and $\e \in [0,1)$, we define $\L(1-\e)$ as the LP obtained by multiplying the right-hand side of the covering constraints by $(1-\e)$, but keeping the packing constraints unchanged.
	\end{definition}

	Notice that the feasible solutions for $\L(1-\e)$ are exactly the $\e$-feasible solutions for $\L$. Now we formally define the stability property.
	
	\begin{definition}[Stability] \label{def:stability}
     A \PCMC LP $\L$ is called \emph{$(\e, \sigma)$-stable} if any optimal $\e$-feasible solution for it has value at most $(1 + \sigma \e) \cdot \OPT(\L)$, i.e., $\OPT(\L(1-\e)) \le (1 + \sigma \e) \cdot \OPT(\L)$.
	\end{definition}

\new{	The following lemma gives a natural sufficient condition for an LP to be stable.
	
	\begin{lemma}\label{lemma:sufStable}
		If a \PCMC LP $\L$ is such that $\L(1 + \frac{1}{\sigma})$ is feasible, then $\L$ is $(\e, \sigma)$-stable for every $\e \in (0,1]$.
	\end{lemma}
	
	\begin{proof}
			Consider any $\e \ge 0$. Let $\{\bar{x}^t\}_t$ be a feasible solution for $\L(1 + \frac{1}{\sigma})$ and let $\{\underline{x}^t\}_t$ be an optimal solution for $\L(1-\e)$. Define $\tilde{x}^t = \frac{\e}{\e + 1/\sigma} \bar{x}^t + (1-\frac{\e}{\e + 1/\sigma}) \underline{x}^t$. It is easy to verify that $\{\tilde{x}^t\}_t$ is feasible for $\L$ and has value 
	\begin{align*}
		\sum_{t = 1}^n \pi^t \tilde{x}^t \ge \left(1-\frac{\e}{\e + 1/\sigma}\right) \OPT(\L(1-\e)),
	\end{align*}
	where the inequality uses the non-negativity of the $\pi^t$'s. This implies that $\OPT(\L) \ge (1-\frac{\e}{\e + 1/\sigma}) \OPT(\L(1-\e))$, which is equivalent to $\OPT(\L(1-\e)) \le (1 + \e \sigma) \OPT(\L)$. This concludes the proof. 
	\end{proof}
}

The idea behind our estimator for $\OPT(\L)$ in the random-order model has been used in prior works on \textbf{packing-only} LPs, e.g.,~\cite{DevanurHayes09,feldman,AWY14,MR12}: observe the first $s$ random blocks $(\bpi^1, \A^1, \C^1), \ldots, (\bpi^s, \A^s, \C^s)$ of the LP, and form the \emph{sampled LP} with these blocks and right-hand side scaled by a factor of $\frac{s}{n}$:
	\begin{align}
		\max ~\tsty \sum_{t = 1}^s \bpi^t \bx^t & \notag\\
		  st ~\tsty \sum_{t = 1}^s \A^t \bx^t &\le \frac{s}{n} \cdot b \label{eq:sampledLP}\\
		     \tsty \sum_{t = 1}^s \C^t \bx^t &\ge \frac{s}{n} \cdot d \notag\\
		     \bx^t &\in \fullsimplex^k \ \ \ \forall t \in [s]. \notag
	\end{align}	
	By the random-order assumption, the blocks $(\bpi^t, \A^t, \C^t)$ are an $\frac{s}{n}$-sample of the underlying LP $\L$, and hence should be ``representative'' of the latter. Consequently, computing an optimal ($\e$-feasible) solution for this sampled LP should give a good estimate of $\frac{s}{n} \cdot \OPT(\L)$. 

        \new{In this work, we extend this idea beyond packing-only LPs, and allow for covering constraints as well. This introduces complications, but the stability property defined above allows us to address them.} Let us define sampled LPs in the generality that we need.

\begin{definition}[Restricted LP]
  \label{def:rest-lps}
  Given a \PCMC LP $\L$, and a subset $I$ of $[n]$, the \emph{restricted LP} $\L^I$ is obtained by retaining only the columns of $\L$ that belong to $I$, and setting the right-hand side $\RHS(\L^I)$ to be $\frac{|I|}{n} \RHS(\L)$, namely
\begin{align*}
		\max ~\tsty \sum_{t \in I} \pi^t x^t & \\
		  st ~\tsty \sum_{t \in I} A^t x^t &\le \frac{|I|}{n} \cdot b \tag{$\L^I$}\\
		     \tsty \sum_{t \in I} C^t x^t &\ge \frac{|I|}{n} \cdot d \\
		     x^t &\in \fullsimplex^k \ \ \ \forall t \in I.
	\end{align*}
\end{definition}
	\new{Notice that the restricted LP $\bL^{\{1, \ldots, s\}}$ is exactly the sampled LP of \eqref{eq:sampledLP}. Moreover, since the sequence of random variables $(\bpi^1, \A^1, \C^1), \ldots, (\bpi^n, \A^n, \C^n)$ is exchangeable (i.e., every subsequence of same size has the same joint distribution), $\bL^I$ has the same distribution as the sampled LP $\bL^{\{1, \ldots, s\}}$ for every $s$-subset $I \subseteq [n]$ (only the indices of the variables are different).}

	In the following, we will need a more refined control over the ratio of item profits to \OPT, we define the \emph{width} of a \PCMC LP $\L$ as 
\begin{gather}
  \min\left\{ \min_{t,i,j} \frac{b_i}{a^t_{ij}}, \; \min_{t,i,j} \frac{d_i}{c^t_{ij}} \right\}~~. \label{eq:width}
\end{gather}
Note that unlike the generalized width defined in (\ref{eq:gw}), the width does not
depend on $\frac{\pi^t_{i,j}}{\OPT}$. Finally, for a vector $v \in \R^n$ and a subset $I \subseteq [n]$, we use $v|_I = (v_i)_{i \in I}$ to denote the vector which keeps only the coordinates in $I$.

	

\new{The following is the main result of this section: The first two parts say that with good probability a solution to the original LP translates to a solution to the sampled LP with value not much lower than the expected value.} This is via a straightforward application of Bernstein's inequality. \new{The latter two parts show that the optimal value for the sampled LP is not much higher either, both in expectation and with high probability.} These parts use $(\e, \sigma)$-stability to control the variance of the (dual of the) sampled LP. Since the proof is slightly technical we defer it to Appendix~\ref{app:samplePC}.

	\begin{lemma}[Sampled LP Lemma] \label{lemma:PCLB}
		For $\e, \delta \in (0,1)$ and $\rho \ge 1$, consider a \PCMC LP $\L$ with width at least $\frac{32 \log(m/\delta)}{\e^2}$ and such that all item values $\pi^t_{i,j}$ are at most $\rho \cdot \OPT(\L) (\frac{\e^2}{32 \log(m/\delta)})$. Let $s \ge \e^2 n$ and let $I$ be a $s$-subset of $[n]$. Then for $\e' = \e \sqrt{\frac{n}{s}}$ we have the following guarantees for the random LP $\bL$:
		
		\begin{enumerate}
                  \setlength{\topsep}{0pt}
                  \setlength{\itemsep}{0pt}
			\item[(a)] If $\bar{\bx}$ is feasible for $\bL$, then  $(1 - \frac{\e'}{2}) \cdot \bar{\bx}|_{I}$ is feasible for $\bL^{I}(1 - \e')$ with probability at least $1 - 2\delta$.

			\item[(b)] With probability at least $1 - 4\delta$, $\OPT(\bL^{I}(1-\e')) \ge (1 - \frac{\rho \e'}{2}) \cdot \frac{s}{n} \, \OPT(\L)$.

		\end{enumerate}
                Moreover, if $\L$ is $(\e, \sigma)$-stable, then the following upper bounds hold:
		\begin{enumerate}
                  \setlength{\topsep}{0pt}
                  \setlength{\itemsep}{0pt}
			\item[(c)] With probability at least $1 - 2\delta$, $\OPT(\bL^{I}(1-\e')) \le (1 + 2 \sigma \e' + \rho \e')\cdot \frac{s}{n} \, \OPT(\L)$

			\item[(d)] $\E\big[\OPT(\bL^{I}(1-\e'))\big] \le  (1 + \sigma \e')\cdot  \frac{s}{n}\, \OPT(\L)$.
		\end{enumerate}
	\end{lemma}


}

\section{Solving packing/covering LPs with unknown OPT}
\label{sec:random-pclps}

\full{
We can now combine the algorithm from \S\ref{sec:lp-to-lb} that needs estimates on $\OPT$ with the algorithm from \S\ref{sec:estPC} that estimates of $\OPT$ to get almost optimal $\e$-feasible online solutions to packing-covering multiple-choice LPs. \new{We are going to use the same \emph{doubling} strategy introduced in \cite{AWY14} (and used in \cite{MR12}) in the packing-only case. We give the details since covering constraints introduce technicalities that preclude applying previous results directly.}
%
%
	To simplify notation we use the shorthand $\bL^{\le k}$ to denote the sampled LP $\bL^{\{1, \ldots, k\}}$, and similarly $\bL^{> k} := \bL^{\{k+1, \ldots, n\}}$.

\subsection{One-time learning OPT}
\label{sec:one-time}
}

\short{
We now turn to the problem of estimating the value of $\OPT$ for a packing-covering problem, so that we can use it in Theorem~\ref{thm:LPAssEstimates}. For that, we will require a kind of \emph{stability property}, loosely asking that the covering constraints are not ``very tight''.

	\begin{definition}[Stability] \label{def:stability}
     A packing-covering LP $\L$ is called \emph{$(\e, \sigma)$-stable} if any optimal $\e$-feasible solution for it has value at most $(1 + \sigma \e) \OPT(\L)$.
	\end{definition}

	For a randomly permuted LP $\bL$, the basic idea to obtain an estimate of $\OPT(\bL)$ is to see the first $n/2$ random blocks and form a \emph{sampled LP} with these blocks and right-hand side scaled by a factor of $1/2$; computing an optimal $O(\e)$-solution for this sampled LP should give a good estimate of $\OPT(\bL)/2$. 

\begin{definition}[Restricted LP]
  \label{def:rest-lps}
  Given a packing-covering LP $\L$, and a subset $I$ of $[n]$, the restricted LP
  $\L^I$ is obtained by retaining only the columns of $\L$ that belong
  to $I$, and setting the right hand side $\RHS(\L^I)$ to be
  $\frac{|I|}{n} \RHS(\L)$.
\end{definition}

	For a $(\e, \sigma)$-stable $\L$, we show that an $O(\e)$-optimal solution for the sampled LP $\bL^{\{1, \ldots, \frac{n}{2}\}}$ is, with probability at least $1-\delta$, is about $\frac{1}{2} \OPT(\L) \pm \sigma \e \OPT(\L)$; while the lower bound is a straightforward application of Bernstein's inequality, for the upper bound we uses $(\e, \sigma)$-stability to control the variance of the (dual of the) sampled LP.

	Given this bound, we can estimate the optimum of the remaining LP $\bL^{>\frac{n}{2}}$ using the sampled LP $\OPT(\bL^{\le\frac{n}{2}})$ and employ Theorem \ref{thm:LPAssEstimates} to the former to obtain the following.
}

	\new{The \emph{one-time learning} (\OTL) algorithm below uses the first half of the instance to estimate the value of $\OPT$ and then runs the algorithm \LPtoLB on the second half.
	
	\begin{algorithm}[H]
		\caption{. One-time learning (\OTL) algorithm}
		\begin{algorithmic}[0]
			\State \textbf{Input:} Reals $\e, \delta, \sigma > 0$, number $n$ of items, and right-hand sides $b$ and $d$.
			\medskip
			
			\State \textbf{(Estimation phase)} Let $\e' = \e \sqrt{2}$. For time steps $1, \ldots, \frac{n}{2}$ observe blocks $(\bpi^1, \A^1, \C^1),$ $\ldots, (\bpi^{n/2}, \A^{n/2}, \C^{n/2})$ and compute the estimate $\estopt := \frac{1-2\e'}{1+ 3 \sigma} \cdot \OPT(\L^{\le \frac{n}{2}}(1-\e'))$. Output $\bar{\bx}^t = 0$ for all these time steps.
			
			\medskip
			\State \textbf{(Optimization phase)} For the remaining time steps $\frac{n}{2} + 1, \ldots, n$, run algorithm \LPtoLB with the computed estimate $\estopt$ over the scaled remaining LP $\bL^{> \frac{n}{2}}(1-\e')$, outputting the solution $\bar{\bx}^{n/2 + 1}, \ldots, \bar{\bx}^n$ returned by \LPtoLB.
		\end{algorithmic}
	\end{algorithm}}
	
	The following theorem shows that the solution computed is almost feasible and almost optimal for the second half $\bL^{>\frac{n}{2}}$ of the instance.
	
\begin{thm}
  \label{thm:one-phasePC}
  Consider a feasible \PCMC LP $\L$. Consider $\e > 0$ at most a sufficiently small constant, and $\delta \in (0,\e]$, and assume that the generalized width of $\L$ is at least $\frac{32 \log(m/\delta)}{\e^2}$. Also assume that $\L$ is $(\e, \sigma)$-stable for $\sigma \in [1, \frac{1}{20 \sqrt{2} \e}]$. Then with probability at least $1 - 13 \delta$, the algorithm \OTL returns an online solution $\{\bar{\bx}^t\}_{t = n/2}^n$ that is $(c_4 \e)$-feasible for $\bL^{> \frac{n}{2}}$ and has value 
  \begin{align*}
  	\sum_{t = 1}^n \bpi^t \bar{\bx}^t \ge (1 - c_5 \sigma \e)\, \OPT\big(\bL^{> \frac{n}{2}}\big(1-\sqrt{2} \e\big)\big)
  \end{align*}
  for constants $c_4, c_5$.
\end{thm}

	The proof is conceptually direct: it uses Lemma~\ref{lemma:PCLB} to assert the quality of the estimate $\estopt$ defined, and then employs Theorem~\ref{thm:LPAssEstimates} to obtain guarantees on the computed solution $\{\bar{\bx}^t\}_{t=n/2+1}^n$. The details are in Appendix \ref{app:otl}. 

Observe that this theorem does not give an almost feasible solution for the entire LP $\bL$, only for $\bL^{> n/2}$. Also it only achieves value $\approx \OPT(\bL^{> n/2}) \approx \frac12 \, \OPT(\L)$. To do better we use the doubling strategy to dynamically learn $\OPT$.


\full{
\subsection{Dynamically learning OPT} \label{sec:multiPhasePC}
}
\short{
	To avoid the total loss of value from the first $n/2$ columns of the LP, instead of applying the above theorem directly to input LP, we apply it to the sub-LPs of doubling sizes $\bL^{\le 2 \e n}$, $\bL^{\le 4 \e n}$, $\bL^{\le 8 \e n}$, etc. to obtain a sequence of ``disjoint'' solutions which, when put together, give an $O(\e)$-feasible solution for the whole LP. The total loss in value is about $\OPT(\bL^{\le \e n}) + \sigma \log \e^{-1} \OPT(\L)$. This leads to the following result.
}
	
	Our final algorithm for \PCMC LPs is called \emph{dynamic learning algorithm} (\DLA). The idea of this algorithm is to use the blocks seen in the first $\e n$ time steps to obtain an estimate of $\OPT(\L)$ and feed it to \LPtoLB to compute solutions for time steps $\e n + 1, \ldots, \e n$. Then it uses the blocks seen up to this time (i.e. the first $2 \e n$ blocks) to obtain a better estimate of $\OPT(\L)$ and again feeds it to \LPtoLB, obtaining solutions for time steps $2 \e n +1, \ldots, 4\e n$; the algorithm proceeds in this doubling fashion until the end. Equivalently, we can define it as repeatedly using algorithm \OTL from the previous section:
	
\new{	
	\begin{algorithm}[H]
		\caption{. Dynamic learning algorithm (\DLA) }
		\begin{algorithmic}[0]
			\State \textbf{Input:} Reals $\e, \delta, \sigma > 0$, number $n$ of items, and right-hand sides $b$ and $d$.
			\medskip
			\State For $i = 0, \ldots, \log \e^{-1}$, let $S_i := \{1, \ldots, 2^i \e n\}$ and let $\e_i = \e \sqrt{n/|S_i|} = \sqrt{\e/2^i}$
			
			\medskip
			\For{$i = 1,\ldots, \log \e^{-1}$}	
				\State (Phase $i$) Run the algorithm \OTL with parameters $\e_i,\delta, \sigma$ over the instance $\bL^{S_i}$, namely the first $2^i \e n$ times steps of instance $\bL$, and let $\big(0, \ldots, 0, \bar{\bx}^{|S_{i-1}| + 1}, \bar{\bx}^{|S_{i-1}| + 2}, \ldots, \bar{\bx}^{|S_i|}\big)$ be the returned solution
			\EndFor
			\medskip
			\State Output $\big(0, \ldots, 0, \bar{\bx}^{|S_0|+1}, \bar{\bx}^{|S_0|+2}, \ldots, \bar{\bx}^n\big)$, namely the ``union'' of the solutions computed above
		\end{algorithmic}
	\end{algorithm}
}

Algorithm \DLA finally delivered the promised guarantees for \PCMC LPs: for stable LPs with not too small generalized width, it obtains an $O(\e)$-feasible $(1-O(\e))$-optimal solution with high probability. (This corresponds to Theorem \ref{ithm:main-II} in the introduction.)

\begin{thm}
  \label{thm:multi-phasePC}
  Consider a feasible \PCMC LP $\L$. Consider $\e > 0$ at most a sufficiently small constant $\e_0$, and $\delta \in (0,\e]$, and assume that the generalized width of $\L$ is at least $\frac{32 \log(m/\delta)}{\e^2}$. Also assume that $\L$ is $(\e, \sigma)$-stable for $\sigma \in [1, \frac{1}{20 \sqrt{2} \e}]$. \new{Then with probability at least $1 - 23 \delta \log(1/\e)$, the algorithm \DLA returns an online solution $\{\bar{\bx}^t\}_t$ that is $(c_6 \e)$-feasible for $\bL$ and has value}
  \begin{align*}
  	\sum_{t = 1}^n \bpi^t \bar{\bx}^t \ge (1 - c_7 \sigma \e)\, \OPT(\L)
  \end{align*}
  for constants $c_6, c_7$.
\end{thm}

\full{

		The main element for proving this theorem is to obtain guarantees for the solutions returned by \OTL in each phase $i$ of the algorithm \DLA. This is done by showing that with good probability the random instance $\bL^{S_i}(1-\e_i)$ is $(\e_i, O(\sigma))$-stable (using the bounds from Lemma \ref{lemma:PCLB}) so we can employ Theorem \ref{thm:one-phasePC} to obtain guarantees for \OTL; a full proof is presented in Appendix \ref{app:phaseDLA}.
		
		\begin{lemma} \label{lemma:phaseDLA}
			Let $\bar{\bx}(i) = \big(0, \ldots, 0, \bar{\bx}^{|S_{i-1}| + 1}, \bar{\bx}^{|S_{i-1}| + 2}, \ldots, \bar{\bx}^{|S_i|}\big)$ be the partial solution computed in phase $i$ of algorithm \DLA. Then with probability at least $1-19\delta$, $\bar{x}(i)$ is $((c_4 + 1) \e_i)$-feasible for $\bL^{S_i \setminus S_{i-1}}$ and has value at least $\sum_{t = 1}^{|S_i|} \bpi^t \bar{\bx}(i)^t \ge (1- 7 c_5 \sigma \e_i) \, \OPT(\bL^{S_i \setminus S_{i-1}}(1-\sqrt{2} \e_i))$. 
		\end{lemma}
		
	We now use this lemma on every phase of \DLA to obtain the proof of Theorem \ref{thm:multi-phasePC}.
		
	\begin{proof}[Proof of Theorem \ref{thm:multi-phasePC}]
		Let $\cE$ be the event that for all $i = 1, \ldots, \log(\frac{1}{\e})$, the bound from Lemma \ref{lemma:phaseDLA} holds and
		\begin{align*}
			\OPT(\bL^{S_i \setminus S_{i-1}}(1-\sqrt{2} \e_i)) \ge \big(1- \e_i\big) \cdot \frac{|S_i \setminus S_{i-1}|}{n}\, \OPT(\L) = \big(1- \e_i\big)\, \e\, 2^{i-1}\,\OPT(\L)
		\end{align*}
		also holds. Using Lemmas \ref{lemma:PCLB} and \ref{lemma:phaseDLA} and taking a union bound over $i = 1, \ldots, \log(\frac{1}{\e})$, we get that event $\cE$ holds with probability at least $1 - 23 \delta \log (\frac{1}{\e})$ (indeed we can apply Lemma \ref{lemma:PCLB} since $\e_i \ge \e$ implies that the generalized width of $\bL^{S_i \setminus S_{i-1}}$ is at least $\frac{32 \log(m/\delta)}{\e^2} \ge \frac{32 \log(m/\delta)}{\e_i^2}$, and $\e_i^2 \le \e$ implies $|S_i \setminus S_{i-1}| \ge \e_i^2 n$). In addition, notice that the definition of the $\e_i$'s gives 
		\begin{align}
			\sum_{i = 1}^{\log(1/\e)} \e_i 2^{i-1} = \frac{\sqrt{\e}}{2} \sum_{i = 1}^{\log(1/\e)} \sqrt{2^i} \le \frac{\sqrt{\e}}{2} \cdot \frac{\sqrt{2}^{\log(1/\e)+1}}{\sqrt{2} - 1} =  \frac{1}{\sqrt{2}\,(\sqrt{2} - 1)} \le 2. \label{eq:epsSum}
		\end{align}
		
	Then using the guarantees from Lemma \ref{lemma:phaseDLA}, it is easy to see that under the event $\cE$ the solution $\{\bar{\bx}^t\}_t$ output by \DLA satisfies the covering constraints within a factor of $(1-\e(1+2 c_4))$:
	\begin{align*}
		&\sum_{t=1}^n \C^t \overline{\bx}^t \ge d \cdot \sum_{i = 1}^{\log(1/\e)} (1 - (c_4 + 1) \e_i) \frac{|S_i \setminus S_{i-1}|}{n}  = d \cdot \left( 1 - \e - \e (c_4 + 1) \sum_{i=1}^{\log(1/\e)} \e_i 2^{i-1}\right) \ge d \cdot \left( 1 - \e (3 + 2 c_4)\right).
	\end{align*}
	Similar, and even simpler, calculations show that the solution $\{\bar{\bx}^t\}_t$ also satisfies the packing constraints under $\cE$, and thus is $(3 +2 c_4)$-feasible for $\bL$. Finally, since under $\cE$ we have that $\sum_{t = |S_{i-1}| + 1}^{|S_i|} \bar{\bx}^t \ge (1 - 7 c_5 \sigma \e_i) (1- \e_i)\, \e 2^{i-1} \, \OPT(\L)$, adding over all phases $i$ gives
	\begin{align*}
		\sum_{t = 1}^n \bpi^t \bar{\bx}^t &\ge \e \, \OPT(\L) \sum_{i = 1}^{\log(1/\e)} (1 - 7 c_5 \sigma \e_i) \big(1- \e_i\big)\, 2^{i-1} \ge \e \, \OPT(\L) \sum_{i = 1}^{\log(1/\e)} \left(1 - \frac{c_7}{2} \sigma \e_i\right) \, 2^{i-1} \\
			&\ge \OPT(\L) \, (1 - c_7 \e \sigma) 
	\end{align*}
	for some constant $c_7$, where in the second inequality we also used the fact $\sigma \ge 1$. This concludes the proof of the theorem.
	\end{proof}
}


\newcommand{\brho}{\boldsymbol{\rho}}
\newcommand{\bgamma}{\boldsymbol{\gamma}}
\newcommand{\vrho}{\vec{\rho}}
\newcommand{\bOhat}{\widehat{\mathbf{O}}}
\newcommand{\bbhat}{\widehat{\mathbf{b}}}
\newcommand{\pbar}{\overline{p}}
\newcommand{\PLPn}{\PLP_n}

\newcommand{\skim}[2]{\textsf{skim}_{#2}(#1)}
\newcommand{\skimLt}{\skim{\L}{\tau}}

\newcommand{\tmax}{\mathbf{t}_{\textrm{last}}}
\newcommand{\val}{\mathsf{val}}

\newcommand{\constC}{24}

\full{
\section{Solving packing LPs with no assumption on item values}
\label{sec:plp-no-limit}

Finally, we turn to removing all assumptions about the magnitudes of the
profits $\pi_t$. To achieve this, we consider packing-only LPs ({\PLP}s) of the following form:
	\begin{align}
		\max ~&\tsty \sum_{t = 1}^n \pi_t x_t  \tag{\PLP} \label{pcmc} \\
		  st ~&A x \le b \notag\\
		      &x \in [0,1]^n,\notag
	\end{align}
	for non-negative $\pi, A,$ and $b$. Observe the absence of covering constraints, and also of multiple-choice constraints---each load matrix $A^t$ is just
a column vector, so at each step one more column of the matrix $A$ is
revealed to us. In this setting we now find a solution of value
$(1-O(\e))\, \OPT$, without any assumptions about the magnitudes of the values $\pi_t$.

The presence of items with large value $\pi_t$ compared to $\OPT$ rules out using algorithms \DLA or \LPtoLB directly, since we may violate the required generalized width bounds. 
We would like to show a threshold value such 
that items having  values higher than the threshold can be directly added to our solution (i.e., we can set $x_t = 1$ for these items) such that: \new{(i)~the total value of these items plus the LP on the remaining items is still a $(1 - \e)$-approximation to the original optimal value (i.e., we do not lose much value by adding these items); (ii)~the individual values of items in the remaining LP are $O\big(\frac{\e^2}{\log(m/\delta)}\big)\cdot \OPT$, and (iii)~the right-hand side of the remaining LP
is still large. Properties (ii) and (iii) would guarantee that the generalized width of the remaining LP is $\Omega\big(\frac{\log(m/\delta)}{\e^2}\big)$ so we can employ the algorithm \LPtoLB, for example, on the remaining LP. } 

\new{Indeed, our algorithm proceeds essentially using this approach. The difference is that we can only show a threshold such that the remaining items have values smaller than $O\big(\frac{\e^2}{\log(m/\delta)}\big)$ times the \emph{original LP value}, not the new LP value. Thus, having taken high-value items into our solution, the value of the new ``skimmed-off'' LP may have fallen drastically, so we may get no bounds on the generalized width. Yet, the value of the remaining LP being small implies the already-picked items have a lot of value, which means even a poor approximation on the remaining LP suffices. This approach introduces additional complications that are handled in the remainder of \S\ref{sec:plp-no-limit}.
} 

\new{Let us establish some notation.} Denote the set of
{\PLP}s with $n$ columns by $\PLP_n$, and the right-hand side vector $b$ of a \PLP $\L$ by $\RHS(\L)$.
Also, recall that the \emph{width} (defined in~(\ref{eq:width})) of a packing LP is  $\max_{i,j}\{
\frac{b_i}{a_{ij}} \}$. Define $B := \frac{\log ((m + 1)/\delta)}{\e^2}$, and will want
the width of all LPs to be $\Omega(B)$. \emph{By scaling, we will assume that entries in each column $A^t$ lie in $[0,1]$ and that the right-hand sides $b_i$'s are at least the width of the LP.}


\subsection{Estimating the high-value items}
\label{sec:good-exists}

The first piece of the puzzle is to show the existence of a good
skimming threshold $\tau$ such that blindly adding to the solution the items with value more than $\tau$ satisfies the properties (i)--(iii) above. We first formally define the ``remaining LP'' mentioned above. 

\begin{definition}[Skimmed LP]
  \label{def:skim-lps}
  For $\L \in \PLPn$ and a threshold $\tau \geq 0$, consider the high-value items $H = \{t : \pi_t > \tau\}$. The \emph{skimmed LP} $\skimLt$ is then obtained from $\L$ by \new{setting to zero the value of all items in $H$} and setting the right-hand side to be $\RHS(\skimLt) := \RHS(\L) - \sum_{t \in H} A^t$. (I.e., this is essentially the residual LP if we hard-code $x_t = 1$ for all $t$ with $\pi_t > \tau$.)
  
\end{definition}

	The following gives a formal definition of a good skimming threshold (it is useful to think of $\Delta = \Theta(\OPT(\L))$, which corresponds to the discussion above).

\begin{definition}[Good Skimming Threshold] 
  \label{def:goodskim}
  A value $\tau \in \R$ is a \emph{$\Delta$-good skimming threshold} for $\L  \in \PLPn$ if:
  \begin{OneLiners}
  \item[(S1)] $\OPT(\skimLt) + \sum_{t: \pi_t > \tau}
    \pi_t  \ge (1- \e)
    \, \OPT(\L)$,
  \item[(S2)] $\RHS(\skimLt) \geq \frac12 \RHS(\L)$, and
  \item[(S3)] $\tau \le \frac{\constC \cdot \Delta}{B}$.  
  \end{OneLiners}
\end{definition}

The starting point is the following simple but crucial claim, which gives a ``perfect'' skimming threshold $\tau$ --- one where we do not lose
\emph{any value} by blindly picking items of higher value. \new{It is} a generalization of the ``value/weight'' rule for
solving single knapsack problems, and says that every optimal fractional
solution must pick all the sufficiently high valued items. 
	
\begin{lemma}[Perfect Threshold]
  \label{lemma:pickTop}
  For a feasible \PLP $\L$ and for every $\tau \ge \frac{\OPT(\L)}{\min_i
    b_i}$, we have \[ \OPT(\skim{\L}{\tau}) + \sum_{t :
    \pi_t > \tau} \pi_t = \OPT(\L).\]
\end{lemma}
		
\begin{proof}
  The dual of $\L$ is given by 
  \begin{align*}
    \D = \min \left\{ \sum_{i \in [m]} b_i p_i + \sum_{t \in [n]} y_t\ \mid \ 
      \ip{p}{A^t} + y_t \geq \pi_t \;\ \forall t \in [n],
      \ p, y \geq 0 \right\},
  \end{align*}
  where the variable $y_t$ is relative to the primal constraint $x_t \le 1$.
		
  Let $(x^*, (p^*, y^*))$ be an optimal primal-dual pair for $\L$ and
  $\D$. 
  We claim that
  $x^*$ selects all items with $\pi_t > \|p^*\|_1$. To see this,  
  observe that $\pi_t > \ip{p^*}{A^t}$ means $y^*_t > 0$ and hence $x^*_t = 1$ by
  complementary slackness. But since the entries
  of $A^t$ are upper bounded by 1, we have $\ip{p^*}{A^t} 
  \le \|p^*\|_1$. Thus, $\pi_t > \|p^*\|_1
  \ge \ip{p^*}{A^t}$ implies $x^*_t = 1$, giving the claim. 
  
  This
  means that 
  \[ \OPT\big(\skim{\L}{\|p^*\|_1}\big) + \sum_{t : \pi_t > \| p^* \|_1} \pi_t =
  \OPT(\L).\] Now it suffices to show that $\|p^* \|_1 \leq
  \frac{\OPT(\L)}{\min_i b_i}$.  This is easy from strong duality:
    $\OPT(\L) = \OPT(\D) \geq \sum_{i \in [m]} b_i p_i^* \geq (\min_i b_i) \cdot
    \| p^* \|_1$, 
  which completes the proof.
\end{proof}

	\new{Lemma ensures $\frac{\OPT(\L)}{\min_i
    b_i}$ satisfies criteria~(S1) and~(S3) for being a $\Theta(\OPT(\L))$-good skimming threshold notice. But if we assume that $\L$ has width at least $B$, the slightly higher threshold $\tau = \frac{2 \OPT(\L)}{\min_i b_i} \le \frac{2 \OPT(\L)}{B}$ satisfies (S2) as well. Indeed, at most $\frac{B}{2}$ items with value more than $\tau$ (else picking $\frac{B}{2} +1$ of them would satisfy the constraints of the LP $\L$ and give value strictly greater than $\OPT(\L)$), and thus $\RHS(\skim{\L}{\tau}) \ge \RHS(\L) - \frac{B}{2}\,\ones \ge \frac{1}{2} \RHS(\L)$. Therefore, this threshold $\tau$ is indeed a ``perfect'' $\Theta(\OPT(\L))$-good skimming threshold for $\L$.
	
  Hence we want to estimate $\tau$ in order to pick up the high-value items in an online fashion, i.e., we seem to need to estimate $\OPT(\L)$ with good probability (in the presence of unboundedly large item values $\pi_t$). And hence we seem to be back at square one. However, recall that this is only if we want a perfect threshold. By excluding a small number of items with the largest values $\pi_t$, we now show how to find a good-enough skimming threshold.} To make this precise, for $K \in \mathbb{N}$ let $\L_{<K}$ denote the LP obtained by zeroing out the values $\pi_t$ of the $K$ highest valued items (without changing the right hand side).\footnote{We assume items have distinct values for simplicity; one can break ties consistently and get the same result.}


	\new{First we show that the optimal value of this modified LP $\L_{<K}$ can be (over)estimated by using the optimal value of the sampled LP $\bL^{\le \frac{n}{4}}$.}

\begin{lemma} 
  \label{lemma:optMinusTop} 
  For $\L \in \PLPn$ with width at least $16B$ and $\e \leq \frac{1}{32}$, and $K = 128 \e B$,
  \[ \Pr\bigg[ \tsty \OPT(\bL^{\le \frac{n}{4}}) \ge \frac{\OPT(\L_{<K})}{8}  \bigg] \geq 1 -
  6\delta. \]
\end{lemma}
		
\begin{proof}
  Let
  $M$ be the value of the $K^{th}$ largest value of $\pi$. We use Lemma~\ref{lemma:PCLB}(b)
on $\L_{<K}$ (using the fact that all item values in $\L_{<K}$ are at most $M$ and that since there are no covering constraints $\L_{<K}(1-\e') = \L_{<K}$) to get with probability at
  least $1-4\delta$
  %
  \begin{align}
    \OPT(\bL^{\le \frac{n}{4}}) \ge \OPT((\bL_{<K})^{\le \frac{n}{4}}) &\ge \frac{\OPT(\L_{<K})}{4} - 8 M B \e = \frac{\OPT(\L_{<K})}{4} - \frac{K M}{16}.
    \label{eq:sampleMinusTop}
  \end{align}
  Moreover, we claim that $\OPT(\bL^{\le \frac{n}{4}}) \ge (KM)/16$ with probability at least $1-2\delta$. For that, let $\bT$ be the indices of the top $K = 128\e B$ values of $\bpi$ and consider the solution $\widehat{\bx}$ that
  sets for all $t \in [n/4]$, $\widehat{\bx}_t = 1$ if $t \in \bT$ and $\widehat{\bx}_t = 0$ otherwise. We know that $\RHS(\bL^{\le \frac{n}{4}}) = \frac14
  \RHS(\L) \geq 4 B\ones$;
  moreover, this solution $\widehat{\bx}$ can (at most) pick all $K = 128 \e B \leq
  4 B$ elements of $\bT$ (since $\e \leq 1/32$), so $\widehat{\bx}$ is
  feasible for $\bL^{\le \frac{n}{4}}$. 
  The expected number of items picked by $\widehat{\bx}$ is $\E[|\bT \cap [n/4]|] = K/4$, and the Bernstein's inequality Corollary~\ref{cor:multiChernoff} implies $\Pr(|\bT \cap [n/4]| < K/16) \le 2\exp(-
  \e B) \le 2 \frac{\delta}{m}$. Each item with $\widehat{\bx}_t = 1$ has value at least $M$, so with probability at least $1 - 2 \delta$
  \begin{gather}
    \OPT(\bL^{\le\frac{n}{4}}) \geq \sum_{t \in \bT \cap [n/4]} \bpi_t \widehat{\bx}_t \geq \frac{K M}{16}. \label{eq:SampleMinusTopTwo}
  \end{gather}
  Taking a union bound we have that with probability at least $1-6\delta$ inequalities \eqref{eq:sampleMinusTop}
  and~(\ref{eq:SampleMinusTopTwo}) hold simultaneously. When this holds we can take the average of these inequalities, and so with probability at least $1-6\delta$ $\OPT(\bL^{\le \frac{n}{4}}) \ge \frac{\OPT(\L_{<K})}{8}$, which concludes the proof. 
\end{proof}

	\new{Next, the optimal value of the modified LP $\L_{<K}$ can be used to obtain a good skimming threshold. 
	
	\begin{lemma} \label{lemma:thresholdMinusTop} 
  Consider a feasible \PLP $\L$  with width at least $B$. Suppose $\e \le \frac{1}{3}$ and $K$ is at most $\frac{\e}{2}$ times the width of $\L$. Then any $\tau \ge \frac{3 \OPT(\L_{<K})}{B}$ is a 
$\frac{\tau B}{\constC}$-good skimming threshold for $\L$.		
	\end{lemma}}
	
	\begin{proof}
		We first show that condition (1) of a good skimming threshold holds, namely $\OPT(\skim{\L}{\tau}) + \sum_{t : \pi_t > \tau} \pi_t \ge (1-\e)\, \OPT(\L)$. Let $H = \{t : \pi_t > \tau\}$ be the high-value items, and let $T \subseteq [n]$ denote the index set of the $K$ items with largest value $\pi_t$ (the ``top-value'' items).
		
		 Since the width of $\L_{<K}$ is the same as that of $\L$ (which is at least $B$), we have that $\tau \ge \frac{3 \OPT(\L_{<K})}{\min_i \RHS(\L_{<K})_i}$ and thus from Lemma \ref{lemma:pickTop} we have 
		\begin{align*}
			\OPT(\skim{\L_{<K}}{\tau}) + \sum_{t \notin T, t \in H} \pi_t = \OPT(\L_{<k}).
		\end{align*}
		Adding all the top valued items $T$ to both sides gives 
		\begin{align}
			\OPT(\skim{\L_{<K}}{\tau}) + \sum_{t \in H} \pi_t + \sum_{t \in T, t \notin H} \pi_t = \OPT(\L_{<k}) + \sum_{t \in T} \pi_t \ge \OPT(\L). \label{eq:thresholdMinusTop1}
		\end{align}		
		To connect the left-hand side of this expression with $\OPT(\skim{\L}{\tau})$ we claim that 
		\begin{align}
			\OPT(\skim{\L}{\tau}) \ge (1-\e)\, \OPT(\skim{\L_{<K}}{\tau}) + \sum_{t \in T, t \notin H} \pi_t. \label{eq:thresholdMinusTop2}
		\end{align}
		To see this, consider an optimal solution $x^*$ for $\skim{\L_{<K}}{\tau}$ (which only accrues value from items in $\overline{H}$), and construct the vector $\widehat{x}$ by setting $\widehat{x}_t = (1-\e) x^*_t$ for $t \in \overline{H} \cap \overline{T}$, $\widehat{x}_t = 1$ for $t \in \overline{H} \cap T$, and $\widehat{x}_t = 0$ otherwise. Notice that $\widehat{x}$ is feasible for $\skim{\L}{\tau}$: First, $\RHS(\skim{\L}{\tau}) \ge \RHS(\skim{\L_{<K}}{\tau}) - K \ones$, since only the $K$ items in $T$ can be additionally skimmed in $\L$ compared to $\L_{<k}$. Thus, the usage of the solution $\widehat{x}$ is
	\begin{align*}
		 \sum_{t : t \le \tau} A^t \widehat{x}_t &\le (1-\e) \sum_{t : t \le \tau} A^t x^*_t + K \ones \le (1-\e)\, \RHS(\skim{\L_{<K}}{\tau}) + K \ones\\
		 	&\le (1-\e)\,\RHS(\skim{\L}{\tau})  + 2 K \ones \le \RHS(\skim{\L}{\tau}),
	\end{align*}
	where the inequality $2K \ones \le \e\, \RHS(\skim{\L}{\tau})$ used in the last step follows from the fact that $\RHS(\skim{\L}{\tau}) \ge \frac{1}{2}\,\RHS(\L)$ (condition (2) proved below) and  $\RHS(\L) \ge \frac{2 K}{\e} \ones$ (assumption). Therefore, the solution $\widehat{x}$ is feasible for $\skim{\L}{\tau}$. Since it has value equal to the right-hand side of inequality \eqref{eq:thresholdMinusTop2}, it proves the validity of this inequality.
	
	Putting inequalities \eqref{eq:thresholdMinusTop1} and \ref{eq:thresholdMinusTop2} together and using the non-negativity of the $\pi_t$'s gives $\OPT(\skim{\L}{\tau}) + \sum_{t : \pi_t > \tau} \pi_t \ge (1-\e) \OPT(\L)$, and thus condition (1) holds. 
		
	Condition (3) in the definition of good skimming threshold trivially holds, so to conclude the proof we show condition (2), namely $\RHS(\skim{\L}{\tau}) \ge \frac{1}{2} \RHS(\L)$. As argued before, there are at most $\frac{B}{3}$ items in $H \cap \overline{T}$, otherwise we could pick $\frac{B}{3} + 1$ items in this set to obtain a feasible solution for $\L_{<K}$ (since $A^t \le \ones$ for all $t$) with value more than $\OPT(\L_{<K})$, a contradiction. Therefore, the right-hand side of $\skim{\L}{\tau}$ equals
		\begin{align*}
			\RHS(\L) - \sum_{t \in H} A^t \ge \RHS(\L) - \left(|H \cap \bar{T}| + |T|\right)\cdot\ones \ge \RHS(\L) - \left(\frac{B}{3} + K \right)\cdot \ones.
		\end{align*}
		Since the width of $\L$ is at least $B$ and at least $\frac{2 K}{\e}$, this inequality gives $\skim{\L}{\tau} \ge \RHS(\L) (1 - \frac{1}{3} - \frac{\e}{2}) \ge \frac{1}{2} \RHS(\L)$, where the last inequality uses $\e \le \frac{1}{3}$. This concludes the proof of the lemma.
	\end{proof}	

	Therefore, putting Lemmas \ref{lemma:optMinusTop} and \ref{lemma:thresholdMinusTop} together, we get that with probability at least $1-6\delta$ the value $\btau := \frac{\constC}{B} \, \OPT(\bL^{\le \frac{n}{4}})$ is an $\OPT(\L)$-good skimming threshold for $\L$; in fact, as will be more useful, it is an $\OPT(\L)$-good skimming threshold for $\bL^{>\frac{n}{4}}$ and $\bL^{>\frac{n}{2}}$ (this simply uses the additional fact that $\OPT(\L_{<K})$ is at least $\OPT((\bL^{>\frac{n}{4}})_{<K})$ and $\OPT((\bL^{>\frac{n}{2}})_{<K})$).
	
\begin{lemma}[Finding Good Skimming Threshold]
  \label{lemma:goodThreshold} 
  Consider $\L \in \PLPn$ with width at least $512 B$, $\e \leq \frac{1}{32}$, and $\delta > 0$. Define $\btau := \frac{\constC}{B}
  \OPT(\bL^{\le \frac{n}{4}})$. Then $\btau$ is an $\OPT(\L)$-good skimming threshold for $\bL^{>\frac{n}{4}}$ and $\bL^{>\frac{n}{2}}$ with probability at least $1-6\delta$.
\end{lemma}


\subsection{Estimating $\OPT$ and right-hand side of a skimmed LP}
\label{sec:est-rhs-opt}

	\new{In order to solve the skimmed LP $\skimLt$ using the algorithm \LPtoLB, we need to obtain estimates for the optimal value of this LP, as well as its right-hand side (which is now random). Crucially, since $\skimLt$ has only ``small-value'' items (and we assume $\L$ has large width), these quantities can be estimated by looking at that sampled LP $\bL^{\le \frac{n}{3}}$ just like we did in \S\ref{sec:estPC} for packing-covering LPs; a full proof is deferred to Appendix \ref{app:estRhsOpt}.}
	

\begin{lemma}
  \label{lem:est-rhs}
  Let $\e > 0$ be at most a sufficiently small constant. Given $\L \in \PLPn$ with width at least $40B$, let $\tau$ be an $\OPT(\L)$-good skimming threshold for $\L$.  Then with probability at least $1 - 4\delta$,
  \begin{align*}
    \RHS\big(\skim{\bL^{\le \frac{n}{3}}}{\tau}\big) \in (1 \pm 3 \e)\,\frac{1}{2}\, \RHS\big(\skim{\bL^{> \frac{n}{3}}}{\tau}\big)
  \end{align*}
  and with probability at least $1-16\delta$,
  \begin{align*}
    \OPT\big(\skim{\bL^{\le \frac{n}{3}}}{\tau}\big) \in (1 \pm 3\e)\,\frac{1}{2} \,\OPT\big(\skim{\bL^{> \frac{n}{2}}}{\tau}\big) \pm 33 \e B \tau.
  \end{align*}
\end{lemma}

\new{Next, we show we can replace this fixed good skimming threshold $\tau$ by the estimated one $\btau$ given by Lemma~\ref{lemma:goodThreshold}. For that, we condition that the set of items that appear on times $1, \ldots, \frac{n}{4}$ is such that $\btau$ is a $\OPT(\L)$-good skimming threshold for $\bL^{>\frac{n}{4}}$, and then apply Lemma \ref{lem:est-rhs} to the LP $\bL^{>\frac{n}{4}}$; taking an average over all such conditionings then gives the following.}

\begin{cor}
  \label{cor:est-lp-params}
    Let $\e > 0$ be at most a sufficiently small constant, and let $\L \in \PLPn$ have width at least $512 B$. Define $\btau := \frac{\constC}{B} \OPT(\bL^{\le \frac{n}{4}})$ as in Lemma~\ref{lemma:goodThreshold}. Then with probability at least $1 - 14\delta$ all of the following hold simultaneously:
	\begin{gather*}
		\btau \textrm{ is a $\OPT(\L)$-good skimming threshold for $\bL^{>\frac{n}{2}}$}\\
    \RHS\big(\skim{\bL^{(\frac{n}{4},\frac{n}{2}]}}{\btau}\big) \in (1 \pm 3 \e)\,\frac{1}{2}\, \RHS\big(\skim{\bL^{> \frac{n}{2}}}{\btau}\big)\\    		
    \OPT\big(\skim{\bL^{(\frac{n}{4},\frac{n}{2}]}}{\btau}\big) \in (1 \pm 3\e)\,\frac{1}{2} \,\OPT\big(\skim{\bL^{> \frac{n}{2}}}{\btau}\big) \pm 33 \e B \btau.
  \end{gather*}
\end{cor}

	\new{Finally, since we only have an \textbf{estimate} of the right-hand side of $\skimLt$, we need to ensure that the algorithm \LPtoLB can handle this. Thankfully this is the case: The following corollary of Theorem~\ref{thm:LPAssEstimates} follows directly from the fact that if we reduce the right-hand side of a \PLP by a factor of $(1-6\e)$, the optimal value reduces by at most a factor of $(1-6\e)$ (since we can scale an original optimal solution by $(1-6\e)$ and obtain at least this much value).}
	
		\begin{cor}[of Theorem \ref{thm:LPAssEstimates}]  \label{cor:LPAssPacking}
     Consider a feasible \PLP $\L$. Consider $\e > 0$ at most a sufficiently small constant $\e_0$, and $\delta \in (0, \e]$, and assume that the generalized width of $\L$ is at least $\frac{2 \log ((1+m)/\delta)}{\e^2}$. Suppose the algorithm \LPtoLB is given an estimate $\estopt \in [\frac{\OPT}{2}, \OPT]$ and an approximation $\widehat{b} \in [(1 - 6\e) b, b]$ to the packing right-hand side.
     Then \LPtoLB finds an online solution $\{\bx^t\}_t$ that with probability at least $1- \delta$ is feasible for $\bL$ and has value $\sum_t \bpi^t \bx^t \geq (1 - c_3 \e - 6\e) \,\estopt$.
		\end{cor}


\subsection{Modified one-time learning algorithm}
	
	Now we present the \emph{modified one-time learning} (\mOTL) algorithm, which at a high-level does the natural thing: 1) uses the first $\frac{n}{2}$ time steps (i.e. $\bL^{\le\frac{n}{2}}$) to estimate the required parameters (good skimming threshold $\btau$, optimal value and right-hand side of $\skim{\bL^{> \frac{n}{2}}}{\btau}$); 2) in time steps $> \frac{n}{2}$, picks all high-value items $\bpi_t > \btau$ and feeds the remaining LP $\skim{\bL^{\frac{n}{2}}}{\btau}$ to algorithm \LPtoLB. \new{Here is a formal description of the algorithm \mOTL, which has one small additional feature to simplify its analysis: we truncate the solution obtained for $\skim{\bL^{> \frac{n}{2}}}{\btau}$ to ensure it is always feasible for estimated right-hand side $\bbhat$.

	\begin{algorithm}[H]
		\caption{. Modified one-time learning (\mOTL) algorithm}
		\begin{algorithmic}[0]
			\State \textbf{Input:} Reals $\e, \delta, \sigma > 0$, number $n$ of items, and right-hand side $b$.
			\medskip
			
			\State \textbf{(Estimation)} For time steps $1, \ldots, \frac{n}{2}$ observe blocks $(\bpi^1, \A^1, \C^1),$ $\ldots, (\bpi^{n/2}, \A^{n/2}, \C^{n/2})$ and compute the estimates 
			\begin{align*}
				\btau &:= \frac{\constC}{B}\, \OPT(\bL^{\le\frac{n}{4}})\\
				\bOhat &:= \frac{2}{1+3\e}\, \left[ \OPT(\skim{\bL^{(\frac{n}{4}, \frac{n}{2}]}}{\btau}) - 33 \e B \btau\right] & \textrm{(estimate of $\OPT(\skim{\bL^{> \frac{n}{2}}}{\btau})$)}\\
				\bbhat &:= \frac{2}{1+3\e}\, \RHS(\skim{\bL^{(\frac{n}{4}, \frac{n}{2}]}}{\btau}) & \textrm{(estimate of $\RHS(\skim{\bL^{> \frac{n}{2}}}{\btau})$)}
			\end{align*}
			
			\medskip
			\State \textbf{(Approximate solution for $\skim{\bL^{> \frac{n}{2}}}{\btau}$)} Define $\e' = \e \, \max \Big\{1, \sqrt{2 B \btau/\bOhat}\Big\}$. Run algorithm \LPtoLB with the computed estimates $\bOhat$ and $\bbhat$, and parameters $\e', \delta$, over the skimmed remaining LP $\skim{\bL^{> \frac{n}{2}}}{\btau}$, letting $\tilde{\bx}_{\frac{n}{2} + 1}, \ldots, \tilde{\bx}_n$ be the returned solution. Let $\tmax$ be the largest index such that $\sum_{t = \frac{n}{2} + 1}^{\tmax} \A^t (1-6\e) \tilde{\bx}_t \le \bbhat$.
			
			\bigskip
			\State \textbf{(Output above solution + high-value items, truncated)} Let $\bH = \{t > \frac{n}{2} : \bpi_t > \btau\}$ be the high-value items within the time periods $\frac{n}{2} + 1, \ldots, n$. Output the solution $\bar{\bx}_{\frac{n}{2} + 1}, \ldots, \bar{\bx}_n$ defined by 
			\begin{align*}
				\bar{\bx}_t = \left\{ 
					\begin{array}{ll}
						1 & \textrm{if } t \in \bH\\
						(1-6\e)\,\tilde{\bx}_t & \textrm{if } t \notin \bH, t \in \big[\frac{n}{2} + 1, \tmax\big]\\
						0 & \textrm{if } t \notin \bH, t > \tmax.
					\end{array}
				\right.
			\end{align*}
		\end{algorithmic}
	\end{algorithm}}

	The next theorem states the guarantees that we have for algorithm \mOTL.

\begin{thm}
  \label{thm:one-phase}
  Consider $\e > 0$ smaller than a sufficiently small constant, and $\delta \in (0,\e]$. Consider a feasible LP $\L \in \PLPn$ with width at least $512 B$. Then with probability at least $1- 15\delta$, the algorithm \mOTL returns a solution $\bar{\bx}_{\frac{n}{2} + 1}, \ldots, \bar{\bx}_n$ that is feasible for the LP $\bL^{> \frac{n}{2}}$ and has value at least
  \[ \OPT(\bL^{> \frac{n}{2}}) - c_6 \e\, \OPT(\L) \] for some constant $c_6$.
\end{thm}	

	\new{Here is some intuition for the proof. Recall that the main point of the skimming operation is to try to control the \emph{generalized} width of the skimmed LP $\skim{\bL^{>\frac{n}{2}}}{\btau}$; let us see what we have accomplished in this direction. Suppose (via Corollary \ref{cor:est-lp-params}) that the estimate $\btau$ computed is an $\OPT(\L)$-good skimming threshold for $\bL^{>\frac{n}{2}}$. Property (2) of a good skimming threshold implies that the \emph{width} of $\skim{\bL^{>\frac{n}{2}}}{\btau}$ is the desired $\Omega(B) = \Omega\big(\frac{\log (m/\delta)}{\e^2}\big)$ (if we assume that the width of $\L$ is $\Omega(B)$). However, for the item values, notice that the ratio of
	 $\OPT(\skim{\bL^{>\frac{n}{2}}}{\btau})$ divided by the largest item value in $\skim{\bL^{>\frac{n}{2}}}{\btau}$ is at least 
	 $\frac{\OPT(\skim{\bL^{>\frac{n}{2}}}{\btau})}{\btau} = \Omega(B)\, \frac{\OPT(\skim{\bL^{>\frac{n}{2}}}{\btau})}{\OPT(\L)} = \Omega\big(\frac{\log(m/\delta)}{\e^2}\big) \frac{\OPT(\skim{\bL^{>\frac{n}{2}}}{\btau})}{\OPT(\L)}$.
	 If the skimmed LP has large optimal value (namely $\Omega(\OPT(\L))$), then the skimmed LP has generalized width $\Omega\big(\frac{\log(m/\delta)}{\e^2}\big)$ and we can guarantee that \LPtoLB returns an $(1-O(\e))$-approximation for it; since $\btau$ is a good skimming threshold, this $(1-O(\e))$-approximation together with the high-value items guarantee that algorithm \mOTL obtains value $\approx \OPT(\bL^{>\frac{n}{2}})$.	 	 
	 
	 However, if the skimmed LP has a small optimal value, the generalized width is \emph{not} within the desired $\Omega\big(\frac{\log(m/\delta)}{\e^2}\big)$ range and we lose the guarantee for algorithm \LPtoLB.
	   But the generalized width of $\skim{\bL^{>\frac{n}{2}}}{\btau}$ is always in the range $\Omega\big(\frac{\log(m/\delta)}{(\e'')^2}\big)$ for $\e'' = \Theta(\e) \sqrt{\frac{\OPT(\L)}{\OPT(\skim{\bL^{>\frac{n}{2}}}{\btau})}}$, and thus we can get a $(1-O(\e''))$-approximation (notice the similarity of $\e''$ and the $\e'$ used in the algorithm). Moreover, when $\e'' \gg \e$ this means that the high-value items of $\bL^{>\frac{n}{2}}$ have a large contribution to the optimal value $\OPT(\L)$; since these items are picked by our algorithm, our solution should again have large overall value.}



\begin{proof}[Proof of Theorem \ref{thm:one-phase}]
	Let $\sigma$ be a realization of $\bL^{\le \frac{n}{2}}$ where the guarantees from Corollary \ref{cor:est-lp-params} hold (the probability of obtaining such realization is at least $1-14\delta$). In this scenario, the solution $\bar{\bx}$ output by \mOTL is feasible for $\bL^{> \frac{n}{2}}$: we have $\bbhat \le \RHS(\skim{\bL^{> \frac{n}{2}}}{\btau})$ and thus the choice of stopping time $\tmax$ makes $((1-6\e) \tilde{\bx}_{\frac{n}{2} + 1}, \ldots, (1-6\e) \tilde{\bx}_{\tmax}, 0, \ldots, 0)$ feasible for $\RHS(\skim{\bL^{> \frac{n}{2}}}{\btau})$; the definition of skimmed LP then guarantees that the output solution $\bar{\bx}$ (which consists of these $\tilde{\bx}_t$'s plus high-value items) fits the budgets of $\bL^{> \frac{n}{2}}$, thus we get feasibility.
	
	As for the value $\val$ of the solution under scenario $\sigma$, we have two different cases depending on the value of $\OPT(\skim{\bL^{> \frac{n}{2}}}{\btau})$; let $C$ be a sufficiently large constant ($C = \frac{132 \times 48 \times 2}{1 - 9\e}$ suffices).
	
	\paragraph{Case 1:} \new{$\OPT(\skim{\bL^{> \frac{n}{2}}}{\btau}) < \e C \OPT(\L)$}. Since $\btau$ is a $\OPT(\L)$-good skimming threshold for $\bL^{>\frac{n}{2}}$, the high-value items already give value
	\begin{align*}
		\sum_{t \in \bH} \bpi_t &\ge (1-\e) \OPT(\bL^{>\frac{n}{2}}) - \OPT(\skim{\bL^{> \frac{n}{2}}}{\btau}) \ge
		 \OPT(\bL^{>\frac{n}{2}}) - \e (C + 1) \OPT(\L),
	\end{align*}
	and thus $\val \ge \OPT(\bL^{>\frac{n}{2}}) - \e (C + 1) \OPT(\L)$.
	
	\paragraph{Case 2:} \new{$\OPT(\skim{\bL^{> \frac{n}{2}}}{\btau}) \ge \e C \OPT(\L)$}. We claim that the requirements of Corollary \ref{cor:LPAssPacking} are fulfilled (with approximation parameter $\e'$), so we can apply it to obtain guarantees for the solution $\tilde{\bx}$ for $\skim{\bL^{>\frac{n}{2}}}{\btau}$; most conditions can be easily checked, we argue about the less trivial ones. First, since Corollary \ref{cor:est-lp-params} holds for scenario $\sigma$, $\bOhat \le \OPT\big(\skim{\bL^{>\frac{n}{2}}}{\btau} \big)$; then $\e' \ge \e \sqrt{2 B \btau / \OPT\big(\skim{\bL^{>\frac{n}{2}}}{\btau} \big)}$ and reorganizing this expression gives that the LP $\skim{\bL^{>\frac{n}{2}}}{\btau}$ has generalized width at least $\frac{2 \log((m+1)/\delta)}{(\e')^2}$. Second, the estimate $\bOhat$ is at least $\frac{\OPT(\skim{\bL^{>\frac{n}{2}}}{\btau})}{2}$:
	\begin{align*}
		\bOhat \ge \frac{1-3\e}{1+3\e}\,\OPT(\skim{\bL^{> \frac{n}{2}}}{\btau}) - O(\e B\btau) \ge \frac{1-3\e}{1+3\e}\,\OPT(\skim{\bL^{> \frac{n}{2}}}{\btau}) - O(\e \OPT(\L)) \ge \frac{1}{2} \,\OPT(\skim{\bL^{> \frac{n}{2}}}{\btau}), 
	\end{align*}
	where the last inequality uses the assumption that $\OPT(\skim{\bL^{> \frac{n}{2}}}{\btau}) \ge \e C \OPT(\L)$ for a sufficiently large $C$ and that $\e$ is at most a sufficiently small constant. Finally, using this fact plus the bound $\btau \le \frac{\constC \OPT(\L)}{B}$, we have $\e' \le \e_0$:
	\begin{align*}
		\e' = \e \max\left\{1, \sqrt{2 B \btau/\bOhat} \right\} \le \e \max\left\{1, \sqrt{\frac{4 \times \constC \OPT(\L)}{\OPT(\skim{\bL^{>\frac{n}{2}}}{\btau})}} \right\} \le \max\left\{\e, O(\sqrt{\e}) \right\} \le \e_0,
	\end{align*}
	where the second inequality follows from our assumption on Case 2, and the last inequality follows from the fact $\e$ is a sufficiently small constant.
	
	Thus, we can apply Corollary \ref{cor:LPAssPacking} to $\bL^{>\frac{n}{2}}$ and get that, conditioned on $\sigma$, with probability at least $1-\delta$ the solution $\tilde{\bx}_{\frac{n}{2} + 1}, \ldots, \tilde{\bx}_n$ is feasible for $\bL^{>\frac{n}{2}}$ and has value 
	\begin{align*}
		\sum_{t > \frac{n}{2} : t \notin \bH} \bpi_t \tilde{\bx}_t &\ge (1-c_3 \e' - 6\e')\, \bOhat \ge \bOhat - \e (c_3 + 6) \max\Big\{\bOhat, \sqrt{2 B \btau \bOhat}\Big\} \ge \bOhat - O(\e)\, \OPT(\L)\\
		&\ge \OPT(\skim{\bL^{>\frac{n}{2}}}{\btau}) - O(\e)\,\OPT(\L),
	\end{align*}	
	where the second inequality because $\bOhat \le \OPT(\L)$ (a consequence of an upper bound on $\bOhat$ proved above) and the fact $\btau$ is $\OPT(\L)$-good skimming threshold, and the last inequality uses the above upper bound on $\bOhat$. Also, the feasibility of this solution also implies (together with the conditioning on $\sigma$) that $$\sum_{t > \frac{n}{2}} \A^t \tilde{\bx}_t \le \RHS(\skim{\bL^{>\frac{n}{2}}}{\btau}) \le \frac{1}{1-6\e}\, \bbhat,$$ and hence the stopping time $\tmax$ is set to be equal to $n$. Thus, the value $\val$ obtained by the algorithm uses the full solution $\tilde{\bx}_{\frac{n}{2} + 1}, \ldots, \tilde{\bx}_n$, and again using the fact $\btau$ is a good skimming threshold we get 
	\begin{align*}
		\val &= \sum_{t > \frac{n}{2} : t \notin \bH} \bpi_t \tilde{x}_t + \sum_{t \in \bH} \bpi_t \ge \OPT(\bL^{>\frac{n}{2}}) - O(\e)\,\OPT(\L).
	\end{align*}
	
	Since this all happened with probability at least $1-\delta$ conditioned on a scenario $\sigma$ where Corollary \ref{cor:est-lp-params} holds, taking an expectation over all such scenarios gives that the above guarantee for $\val$ holds with overall probability at least $1-15\delta$. This concludes the proof of the theorem.
	\end{proof}
	

	\subsection{Final algorithm: Modified Dynamic Learning}
	
	\new{The final algorithm to solve packing LPs, \emph{modified dynamic learning} (\mDLA), is exactly as the dynamic learning algorithm from \S\ref{sec:multiPhasePC} with the only change that it uses the modified One-time learning algorithm \mOTL instead of \OTL.}
		
	The following theorem shows that this algorithm \mDLA has the desired guarantees. 


\begin{thm}
  \label{thm:PLP-main}
  
   Consider $\e > 0$ smaller than a sufficiently small constant, and $\delta \in (0,\e]$. Consider a feasible LP $\L \in \PLPn$ with width at least $1024 \, \frac{\log ((m+1)/\delta)}{\e^2}$. 
   \new{Then there is an event $\cE$ that holds with probability at least $1- c_7\,\delta \log \frac{1}{\e}$ such that under $\cE$ the solution $\{\bar{\bx}_t\}_t$ returned by algorithm \mDLA is feasible for the LP $\bL$ and has expected value conditioned on $\cE$ at least $(1- c_8\delta \log \e^{-1})\, \OPT(\L)$, for some constants $c_7, c_8$. In particular, the truncated solution $\bar{\bx}_1, \bar{\bx}_2, \ldots, \bar{\bx}_{\tmax'}, 0, \ldots, 0$, where $\tmax'$ is the largest index that makes this solution feasible for $\bL$, has expected value at least $(1-c_9 \delta \log e^{-1}) \,\OPT(\L)$ for some constant $c_9$.}
\end{thm}

	The proof of this theorem is similar to the proof of the guarantee for the algorithm $\DLA$ and the changes required ---which arise from the fact that we now do not have good bounds on the ratios $\pi_t/\OPT(\L)$--- are indicated in Appendix \ref{app:mDLA}.

}

\short{\noindent\textbf{Note:} \emph{Most proofs, as well as a section proving the
  Informal Theorem~\ref{ithm:packII}, are presented in the final version of the paper.
}}

\ifproc
 \bibliographystyle{splncs03}
 \bibliography{online-lp-short}
\else
 \bibliographystyle{alphaurl}
 {\small \bibliography{online-lp}}

\newcommand{\etalchar}[1]{$^{#1}$}
\begin{thebibliography}{KRTV13}

\bibitem[ABFP13]{packingCovering13}
Yossi Azar, Umang Bhaskar, Lisa Fleischer, and Debmalya Panigrahi.
\newblock Online mixed packing and covering.
\newblock In {\em SODA}, pages 85--100, 2013.

\bibitem[ABH11]{hazanBlackwell}
Jacob Abernethy, Peter~L. Bartlett, and Elad Hazan.
\newblock Blackwell approachability and no-regret learning are equivalent.
\newblock In {\em {COLT} 2011 - The 24th Annual Conference on Learning Theory,
  June 9-11, 2011, Budapest, Hungary}, pages 27--46, 2011.
\newblock URL:
  \url{http://www.jmlr.org/proceedings/papers/v19/abernethy11b/abernethy11b.pdf}.

\bibitem[AD14]{agrawalEC14}
Shipra Agrawal and Nikhil~R. Devanur.
\newblock Bandits with concave rewards and convex knapsacks.
\newblock In {\em Proceedings of the Fifteenth ACM Conference on Economics and
  Computation}, EC '14, pages 989--1006, New York, NY, USA, 2014. ACM.
\newblock URL: \url{http://doi.acm.org/10.1145/2600057.2602844}, \href
  {http://dx.doi.org/10.1145/2600057.2602844}
  {\path{doi:10.1145/2600057.2602844}}.

\bibitem[AD15]{AgrawalD15}
Shipra Agrawal and Nikhil~R. Devanur.
\newblock Fast algorithms for online stochastic convex programming.
\newblock In {\em Proceedings of the Twenty-Sixth Annual {ACM-SIAM} Symposium
  on Discrete Algorithms}, pages 1405--1424, 2015.
\newblock URL: \url{http://dx.doi.org/10.1137/1.9781611973730.93}, \href
  {http://dx.doi.org/10.1137/1.9781611973730.93}
  {\path{doi:10.1137/1.9781611973730.93}}.

\bibitem[AHK12]{AHK12}
Sanjeev Arora, Elad Hazan, and Satyen Kale.
\newblock The multiplicative weights update method: a meta-algorithm and
  applications.
\newblock {\em Theory of Computing}, 8(6):121--164, 2012.
\newblock URL: \url{http://www.theoryofcomputing.org/articles/v008a006}, \href
  {http://dx.doi.org/10.4086/toc.2012.v008a006}
  {\path{doi:10.4086/toc.2012.v008a006}}.

\bibitem[AWY14]{AWY14}
Shipra Agrawal, Zizhuo Wang, and Yinyu Ye.
\newblock A dynamic near-optimal algorithm for online linear programming.
\newblock {\em Operations Research}, 62(4):876--890, 2014.
\newblock URL: \url{http://dx.doi.org/10.1287/opre.2014.1289}, \href
  {http://arxiv.org/abs/http://dx.doi.org/10.1287/opre.2014.1289}
  {\path{arXiv:http://dx.doi.org/10.1287/opre.2014.1289}}, \href
  {http://dx.doi.org/10.1287/opre.2014.1289}
  {\path{doi:10.1287/opre.2014.1289}}.

\bibitem[BIK07]{BIK07}
Moshe Babaioff, Nicole Immorlica, and Robert Kleinberg.
\newblock Matroids, secretary problems, and online mechanisms.
\newblock In {\em SODA '07}, pages 434--443, 2007.

\bibitem[BIKK07]{babaioff}
Moshe Babaioff, Nicole Immorlica, David Kempe, and Robert Kleinberg.
\newblock A knapsack secretary problem with applications.
\newblock In {\em APPROX-RANDOM}, pages 16--28, 2007.

\bibitem[BIKK08]{BabaioffSurvey}
Moshe Babaioff, Nicole Immorlica, David Kempe, and Robert Kleinberg.
\newblock Online auctions and generalized secretary problems.
\newblock {\em SIGecom Exchanges}, 7(2), 2008.

\bibitem[BN09]{buchbinderBook}
Niv Buchbinder and Joseph Naor.
\newblock The design of competitive online algorithms via a primal-dual
  approach.
\newblock {\em Foundations and Trends in Theoretical Computer Science},
  3(2-3):93--263, 2009.

\bibitem[CBL06]{CBL06}
Nicol\`{o} Cesa-Bianchi and G\'{a}bor Lugosi.
\newblock {\em Prediction, Learning, and Games}.
\newblock Cambridge University Press, 2006.

\bibitem[CHW12]{CHW}
Kenneth~L. Clarkson, Elad Hazan, and David~P. Woodruff.
\newblock Sublinear optimization for machine learning.
\newblock {\em J. ACM}, 59(5):23:1--23:49, 2012.

\bibitem[DF80]{diaconis1980}
P.~Diaconis and D.~Freedman.
\newblock Finite exchangeable sequences.
\newblock {\em The Annals of Probability}, 8(4):745--764, 08 1980.
\newblock URL: \url{http://dx.doi.org/10.1214/aop/1176994663}.

\bibitem[DH09]{DevanurHayes09}
Nikhil~R. Devenur and Thomas~P. Hayes.
\newblock The adwords problem: online keyword matching with budgeted bidders
  under random permutations.
\newblock In John Chuang, Lance Fortnow, and Pearl Pu, editors, {\em ACM
  Conference on Electronic Commerce}, pages 71--78. ACM, 2009.

\bibitem[DJSW11]{Devanur11}
Nikhil~R. Devanur, Kamal Jain, Balasubramanian Sivan, and Christopher~A.
  Wilkens.
\newblock Near optimal online algorithms and fast approximation algorithms for
  resource allocation problems.
\newblock In Yoav Shoham, Yan Chen, and Tim Roughgarden, editors, {\em ACM
  Conference on Electronic Commerce}, pages 29--38. ACM, 2011.

\bibitem[Dyn63]{Dynkin}
E.~B. Dynkin.
\newblock {The optimum choice of the instant for stopping a Markov process}.
\newblock {\em Soviet Math. Dokl}, 4, 1963.

\bibitem[ESF14]{UW}
Reza Eghbali, Jon Swenson, and Maryam Fazel.
\newblock Exponentiated subgradient algorithm for online optimization under the
  random permutation model.
\newblock {\em CoRR}, abs/1410.7171, 2014.
\newblock URL: \url{http://arxiv.org/abs/1410.7171}.

\bibitem[FHK{\etalchar{+}}10]{feldman}
Jon Feldman, Monika Henzinger, Nitish Korula, Vahab~S. Mirrokni, and Clifford
  Stein.
\newblock Online stochastic packing applied to display ad allocation.
\newblock In Mark de~Berg and Ulrich Meyer, editors, {\em ESA (1)}, volume 6346
  of {\em Lecture Notes in Computer Science}, pages 182--194. Springer, 2010.

\bibitem[Fre75]{freedman}
David~A. Freedman.
\newblock On tail probabilities for martingales.
\newblock {\em Annals of Probability}, 3:100--118, 1975.
\newblock \href {http://dx.doi.org/10.1214/aop/1176996452}
  {\path{doi:10.1214/aop/1176996452}}.

\bibitem[GK07]{gargKonemann}
Naveen Garg and Jochen K{\"o}nemann.
\newblock Faster and simpler algorithms for multicommodity flow and other
  fractional packing problems.
\newblock {\em SIAM J. Comput.}, 37(2):630--652 (electronic), 2007.
\newblock URL: \url{http://dx.doi.org/10.1137/S0097539704446232}.

\bibitem[GM08]{GoelMehta08}
Gagan Goel and Aranyak Mehta.
\newblock Online budgeted matching in random input models with applications to
  ad{W}ords.
\newblock In Shang-Hua Teng, editor, {\em SODA}, pages 982--991. SIAM, 2008.

\bibitem[Haz06]{hazan}
Elad Hazan.
\newblock Approximate convex optimization by online game playing.
\newblock {\em CoRR}, abs/cs/0610119, 2006.
\newblock URL: \url{http://arxiv.org/abs/cs/0610119}.

\bibitem[Kha04]{K04}
Rohit Khandekar.
\newblock {\em Lagrangian Relaxation Based Algorithms for Convex Programming
  Problems}.
\newblock PhD thesis, Indian Institute of Technology Delhi, March 2004.
\newblock URL:
  \url{http://people.cse.iitd.ac.in/~rohitk/research/thesis.ps.gz}.

\bibitem[Kle05]{kleinberg}
Robert Kleinberg.
\newblock A multiple-choice secretary algorithm with applications to online
  auctions.
\newblock In {\em Proceedings of the sixteenth annual ACM-SIAM symposium on
  Discrete algorithms}, SODA '05, pages 630--631, Philadelphia, PA, USA, 2005.
  Society for Industrial and Applied Mathematics.
\newblock URL: \url{http://portal.acm.org/citation.cfm?id=1070432.1070519}.

\bibitem[KRTV13]{KRTV13}
Thomas Kesselheim, Klaus Radke, Andreas T{\"o}nnis, and Berthold V{\"o}cking.
\newblock Primal beats dual on online packing lps in the random-order model.
\newblock {\em CoRR}, abs/1311.2578, 2013.
\newblock To appear in STOC 2014.
\newblock URL: \url{http://arxiv.org/abs/1311.2578}.

\bibitem[KV05]{KalaiVempala}
Adam Kalai and Santosh Vempala.
\newblock Efficient algorithms for online decision problems.
\newblock {\em Journal of Computer and System Sciences}, 71(3):291 -- 307,
  2005.
\newblock URL:
  \url{http://www.sciencedirect.com/science/article/pii/S0022000004001394},
  \href {http://dx.doi.org/http://dx.doi.org/10.1016/j.jcss.2004.10.016}
  {\path{doi:http://dx.doi.org/10.1016/j.jcss.2004.10.016}}.

\bibitem[MR12]{MR12}
Marco Molinaro and R.~Ravi.
\newblock Geometry of online packing linear programs.
\newblock In {\em Automata, languages, and programming. {P}art {I}}, volume
  7391 of {\em Lecture Notes in Comput. Sci.}, pages 701--713. Springer,
  Heidelberg, 2012.
\newblock URL: \url{http://dx.doi.org/10.1007/978-3-642-31594-7_59}.

\bibitem[MSVV07]{MSVV}
Aranyak Mehta, Amin Saberi, Umesh Vazirani, and Vijay Vazirani.
\newblock Ad{W}ords and generalized online matching.
\newblock {\em J. ACM}, 54(5):Art. 22, 19, 2007.
\newblock URL: \url{http://dx.doi.org/10.1145/1284320.1284321}.

\bibitem[Pru98]{pruss}
Alexander~R. Pruss.
\newblock A maximal inequality for partial sums of finite exchangeable
  sequences of random variables.
\newblock {\em Proceedings of the American Mathematical Society},
  126(6):1811--1819, 1998.

\bibitem[PST95]{tardos}
Serge~A. Plotkin, David~B. Shmoys, and {\'E}va Tardos.
\newblock Fast approximation algorithms for fractional packing and covering
  problems.
\newblock {\em Math. Oper. Res.}, 20(2):257--301, 1995.
\newblock URL: \url{http://dx.doi.org/10.1287/moor.20.2.257}.

\bibitem[Sim89]{simpson}
R.~W. Simpson.
\newblock Using network flow techniques to find shadow prices for market and
  seat inventory control.
\newblock Technical report, 1989.
\newblock MIT Flight Transportation Laboratory Memorandum M89-1.

\bibitem[SS11]{Shalev-ShwartzBook}
Shai Shalev-Shwartz.
\newblock Online learning and online convex optimization.
\newblock {\em Foundations and Trends in Machine Learning}, 4(2):107--194,
  2011.
\newblock URL: \url{http://dx.doi.org/10.1561/2200000018}, \href
  {http://dx.doi.org/10.1561/2200000018} {\path{doi:10.1561/2200000018}}.

\bibitem[TR98]{Talluri1998}
Kalyan Talluri and Garrett~Van Ryzin.
\newblock An analysis of bid-price controls for network revenue management.
\newblock {\em Manage. Sci.}, 44(11):1577--1593, November 1998.
\newblock URL: \url{http://dx.doi.org/10.1287/mnsc.44.11.1577}.

\bibitem[vdVW96]{weakConvergence}
Aad~W. van~der Vaart and Jon~A. Wellner.
\newblock {\em Weak convergence and empirical processes}.
\newblock Springer Series in Statistics. Springer-Verlag, New York, 1996.
\newblock With applications to statistics.
\newblock URL: \url{http://link.springer.com/book/10.1007/978-1-4757-2545-2}.

\bibitem[Wil92]{elWilliamson}
Elizabeth~Louise Williamson.
\newblock {\em Airline network seat inventory control: methodologies and
  revenue impacts}.
\newblock PhD thesis, M.I.T., 1992.

\bibitem[You01]{young}
Neal~E. Young.
\newblock Sequential and parallel algorithms for mixed packing and covering.
\newblock In {\em 42nd {IEEE} {S}ymposium on {F}oundations of {C}omputer
  {S}cience ({L}as {V}egas, {NV}, 2001)}, pages 538--546. IEEE Computer Soc.,
  Los Alamitos, CA, 2001.

\end{thebibliography}

 \appendix

 	\section{Maximal Bernstein for Sampling Without Replacement} \label{app:maximalBernstein}

%
		
		We state the Levy-type maximal inequality for exchangeable random variables obtained in \cite{pruss}. 
		
		\begin{thm}[Theorem 1 of \cite{pruss}, with $\gamma = 2$]\label{thm:pruss}
			Let $Y_1, Y_2, \ldots, Y_n$ be an exchangeable sequence of random variables taking values in a Banach space. Then for every $\lambda \ge 0$ and every $k \le n/2$, we have 
			\begin{align*}
				\Pr\Bigg(\max_{j \le k} \bigg\|\sum_{i \le j} Y_i \bigg\| > \lambda \Bigg) \le 15 \Pr\Bigg(\bigg\|\sum_{i \le k} Y_i\bigg\| > \frac{\lambda}{24} \Bigg).
			\end{align*}
		\end{thm}
		
		\begin{proof}[Proof of Lemma \ref{lemma:maximalBernstein}]
			The sequence $(\sum_{i \le j} X_i - j \mu)_{i = 1}^n$ is exchangeable, so applying the theorem above and then Theorem \ref{thm:conc-hyp} we have 
			\begin{align*}
				&\Pr\Bigg(\max_{j \le k} \bigg|\sum_{i \le j} X_i - j \mu\bigg| > \alpha \Bigg) \le 15 \Pr\Bigg(\bigg|\sum_{i \le k} X_i - k \mu\bigg| > \frac{\alpha}{24} \Bigg) \le 30 \exp \left(- \frac{(\alpha/24)^2}{2 k \sigma^2 + (\alpha/24)} \right).
			\end{align*} 
		\end{proof}

\section{Useful Properties of $(\e, \sigma)$-Stability} 
\label{app:stability}

		For any $\e'$ and subset $I \subseteq [n]$, the dual of $\L^I(1 - \e')$ is 
		\begin{align}
			 \D^I(1 - \e') := \left\{
			 \begin{array}{rl}
			 		\min &\frac{k}{n} \ip{\alpha}{b} - (1 - \e')\frac{k}{n} \ip{\beta}{d} + \sum_{t \in I} \gamma_t \\
			 		& \ip{A^t_{\star, j}}{\alpha} - \ip{C^t_{\star, j}}{\beta} + \gamma_t \ge \pi^t_j ~~~~\forall j, ~~\forall t \in I\\
			 		&\alpha, \beta, \gamma \ge 0
			 \end{array}
			 \right. \label{eq:dualPC}
		\end{align}
	where again we use $A^t_{\star, j}$ to denote the $j$th column of $A^t$, and similarly for $C^t$.
	
	The next essentially follows from the fact that $\eta \mapsto \OPT(\L(1 - \eta))$ is concave, and that its derivative is obtained by looking at an optimal dual solution. 
	
%

	\begin{lemma} \label{lemma:basicStab}
		Let $\L$ be an $(\e, \sigma)$-stable packing/covering LP. Then:
		
		\begin{enumerate}
			\item For every $\e' \ge \e$, $\L$ is $(\e', \sigma)$-stable
						
			\item If $(\tilde{\alpha}, \tilde{\beta}, \tilde{\gamma})$ is an optimal solution for the dual of $\L(1 - \e)$, then $\ip{\bar{\beta}}{d} \le \sigma \cdot \OPT(\L)$.
		\end{enumerate}
	\end{lemma}			
	
	\begin{proof}
		Take an optimal solution $(\tilde{\alpha}, \tilde{\beta}, \tilde{\gamma})$ for $\D(1 - \e)$. Notice that for every $\eta$, $(\bar{\alpha}, \bar{\beta}, \bar{\gamma})$ is feasible for $\D(1 - \e + \eta)$ and hence $\OPT(\D(1 - \e + \eta)) \le \ip{\bar{\alpha}}{b} - (1 - \e + \eta) \ip{\bar{\beta}}{d} + \sum_t \bar{\gamma} = \OPT(\D(1 - \e)) - \eta \ip{\bar{\beta}}{d}$. Using Strong Duality, this gives $\OPT(\L(1 - \e)) \ge \OPT(\L(1 - \e + \eta)) + \eta \ip{\bar{\beta}}{d}$.
		
		Setting $\eta = \e$, $(\e, \sigma)$-stability gives that $\ip{\bar{\beta}}{d} \le \sigma \cdot \OPT(\L)$, proving Part 2. Also, setting $\eta = \e - \e'$ we get $\OPT(\L(1 - \e')) \le \OPT(\L(1 - \e)) + (\e' - \e) \sigma \cdot \OPT(\L)$; again using $(\e,\sigma)$-stability of $\L$ we get $\OPT(\L(1 - \e')) \le (1 + \e' \sigma) \OPT(\L)$ for $\e' \ge \e$, concluding the proof of Part 1.
	\end{proof}		


\section{Proof of Lemma \ref{lemma:PCLB}} \label{app:samplePC}

        Since the LP has width at least $\frac{32 \log(m/\delta)}{\e^2}$, by scaling the constraints we assume without loss of generality that its right-hand sides are all exactly equal to $\vartheta:= \frac{32 \log(m/\delta)}{\e^2}$ and the matrices $A^t, C^t$ on the left-hand side have entries in $[0,1]$.
		
		\paragraph{Part (a).} Consider a feasible solution $\bar{\bx}$ for $\bL$. Since for all $i$, $\sum_t \A^t_i \bar{\bx}^t \le \vartheta$, from Bernstein's inequality (Corollary \ref{cor:multiChernoff}) we have 
		\begin{gather*}
			\Pr\left(\sum_{t \in I} \A^t_i \bar{\bx}^t \ge \frac{s}{n}\vartheta + \tau \vartheta\right) \le 2 \exp \left(- \min \left\{\frac{\tau^2 \vartheta^2}{8 s\vartheta/n}, \frac{\tau \vartheta}{2} \right\}\right) \le 2 \exp \left(- \frac{\log(m/\delta)}{\e^2} \min \left\{\frac{4 \tau^2 n}{s}, 16 \tau \right\}\right);
		\end{gather*}
	setting $\tau = (\e/2) \sqrt{s/n}$ makes the right-hand side at most $2\delta/m$ (this uses the fact $\e \le 1$ and $s \ge \e^2 n$). Then taking a union bound over all $i \in [m_p]$, we have $\sum_{t \in I} \A^t \bar{\bx}^t \le (1 + \frac{\e'}{2}) \frac{s}{n}\vartheta \ones$ with probability at least $1 - 2 \delta (m_p/m)$.
	
	Similarly, fix $i \in \{m_p + 1, \ldots, m\}$ and let $\mu = \E[\sum_{t \in I} \C^t_i \bar{\bx}^t]$; then using the fact $\mu \ge (s/n) \vartheta$, 
	\begin{gather*}
		\Pr\left(\sum_{t \in I} \C^t_i \bar{\bx}^t \le \mu - \tau \mu \frac{n}{s} \right) \le 2 \exp \left(- \min \left\{\frac{\tau^2 \mu (n/s)^2}{8}, \frac{\tau \mu (n/s)}{2} \right\}\right) \le 2 \exp \left(- \frac{\log(m/\delta)}{\e^2} \min \left\{ \frac{4 \tau^2 n}{s}, 16 \tau \right\}\right).
	\end{gather*}	
	Setting $\tau = (\e/2) \sqrt{s/n}$ and taking a union bound over all such $i$'s, we get $\sum_{t \in I} \C^t \bar{\bx}^t \ge (1 - \frac{\e'}{2}) \frac{s}{n}\vartheta \ones$ with probability at least $1 - 2 \delta (m_c/m)$. 
	
	Then, using the fact that $(1 - \frac{\e'}{2}) (1 + \frac{\e'}{2}) \le 1$, with probability at least $1 - 2 \delta$ we have that $(1 - \frac{\e'}{2}) \bar{\bx}|_{I}$ is $(1 - \frac{\e'}{2})^2$-feasible for $\bL^{I}$, and since $(1 - \frac{\e'}{2})^2 \ge (1 - \e')$, Part (a) follows. 

	
	\paragraph{Part (b).} Let $\bar{\bx}$ be an optimal solution for $\L$. Then again by Bernstein's inequality (Corollary \ref{cor:multiChernoff})
	\begin{gather*}
		\Pr\left(\sum_{t \in I} \bpi^t \bar{\bx}^t \le \frac{s}{n} \OPT(\L) - \tau \OPT(\L)\right) \le 2 \exp \left(- \min \left\{\frac{\tau^2 \OPT(\L)^2 \vartheta}{8 \rho (s/n) \OPT(\L)^2 }, \frac{\tau \OPT(\L) \vartheta}{2 \rho (s/n) \OPT(\L)} \right\}\right);
	\end{gather*}
	setting $\tau = (\e \rho/2) \sqrt{(s/n)}$ makes the right-hand side at most $2 \delta$ (this uses the fact $\rho \ge 1$).
	
	Then using Part (a), with probability at least $1 - 4\delta$ we have $\bar{\bx}|_{I}$ feasible for $\bL^{I}(1 - \e')$ and with value at least $(1 - \frac{\rho \e'}{2}) \frac{s}{n} \OPT(\bL)$; this lower bounds the value of $\OPT(\bL^{I}(1-\e'))$ and Part (b) follows.
	




	\medskip

	Now for the arguments about the upper bounds. Here we use the $(\e, \sigma)$-stability of $\L$. We consider the (random) dual $\bD^I$ of $\bL^I$ (recall from \eqref{eq:dualPC} the expression for the dual). 
	
	\paragraph{Part (d).} Let $(\bar{\balpha}, \bar{\bbeta}, \bar{\bgamma})$ be an optimal solution for $\bD(1 - \e)$. Notice that in every scenario this solution is feasible for $\bD^{I}(1 - \e')$, so taking expectations and using Strong Duality $\E[\bD^{I}(1 - \e')] \le \frac{s}{n} \OPT(\L(1-\e)) + (\e' - \e) \frac{s}{n} \ip{\bar{\bbeta}}{\vartheta \ones}$. Using $(\e, \sigma)$-stability of $\bL$ we have $\OPT(\L(1 - \e)) \le (1 + \sigma\e) \OPT(\L)$ and using Lemma \ref{lemma:basicStab} $\ip{\bar{\bbeta}}{\vartheta \ones} \le \sigma \cdot \OPT(\L)$. Since $\e' \ge \e$, this gives $\E[\bD^{I}(1 - \e')] \le \frac{s}{n} (1 + \sigma\e') \OPT(\L)$, and Strong Duality then gives Part (e) of the lemma.
		
	\paragraph{Part (c).} Again let $(\bar{\balpha}, \bar{\bbeta}, \bar{\bgamma})$ be an optimal solution for $\bD(1 - \e)$. Notice that the above argument gives that in every scenario $\OPT(\L^{I}(1 - \e')) \le \frac{s}{n} \OPT(\L(1-\e)) + (\e' - \e) \sigma \cdot \OPT(\L) + (\sum_{t \in I} \bar{\bgamma}_t - \E[\sum_{t \in I} \bar{\bgamma}_t]) \le \frac{s}{n} (1 + \sigma \e') \OPT(\L) + (\sum_{t \in I} \bar{\bgamma}_t - \E[\sum_{t \in I} \bar{\bgamma}_t])$. Thus it suffices to show that with probability at least $1 - \delta$, $\sum_{t \in I} \bar{\bgamma}_t \le \E[\sum_{t \in I} \bar{\bgamma}_t] + \e' \frac{s}{n} (\rho + \sigma) \OPT(\L)$.
		
		For that, we need to control how large $\bar{\bgamma}$ can be. By the optimality of $\bar{\bgamma}$, notice that in each scenario  and for all $t$ we have $\bar{\bgamma}_t \le \max_j \big\{\bpi^t_j, \bip{\bar{\bbeta}}{C^t_{\star,j}}\big\}$, or else we could reduce $\bar{\bgamma}_t$ and obtain a feasible solution for $\bD(1 - \e)$ with strictly smaller cost. Again from Lemma \ref{lemma:basicStab} we have $\sum_i \bar{\bbeta}_i \le \sigma \cdot \OPT(\L)/\vartheta$, and since the entries of $\C^t$ are at most 1, we get $\bip{\bar{\bbeta}}{\C^t_{\star,j}} \le \sigma \cdot \OPT(\L)/\vartheta$. Using our assumption that $\bpi^t_j \le \rho \OPT(\L)/\vartheta$, we get $\bar{\bgamma}_t \le (\sigma + \rho) \OPT(\L)/\vartheta$.
		
		In addition, we claim that $\E[\sum_{t \in I} \bar{\bgamma}_t] \le \frac{s}{n} (1 + \sigma) \OPT(\L)$. Since $\OPT(\L(1-\e)) = \OPT(\bL(1-\e)) = \ip{\bar{\balpha}}{\vartheta \ones} - (1 - \e) \ip{\bar{\bbeta}}{\vartheta \ones} + \sum_t \bar{\bgamma}_t$, reorganizing we get $\sum_t \bar{\bgamma}_t \le \OPT(\L(1 - \e)) + (1 - \e) \ip{\bar{\bbeta}}{\vartheta \ones}.$ Using $(\e, \sigma)$-stability of $\L$ and again the bound from Lemma \ref{lemma:basicStab}, this gives $\sum_t \bar{\bgamma}_t \le (1 + \sigma) \OPT(\L)$, and the claim follows. 
		
		Then using Bernstein's inequality (Corollary \ref{cor:multiChernoff}), we get (letting $\mu = \E\left[\sum_{t \in I} \bar{\bgamma}_t\right]$)
	\begin{align*}
		&\Pr\left(\sum_{t \in I} \bar{\bgamma}_t \ge \mu + \e \tau \OPT(\bL) \right) \le 2 \exp \left(- \min \left\{\frac{\e^2 \tau^2 \OPT(\L)^2  \vartheta}{8 \mu (\rho + \sigma) \OPT(\L)}, \frac{\e \tau \OPT(\L) \vartheta}{2 (\rho + \sigma) \OPT(\L)} \right\}\right) \\
		&\le 2 \exp \left(- \frac{\log(1/\delta)}{(\rho + \sigma)} \min \left\{\frac{4 \tau^2}{(s/n) (1+\sigma)}, \frac{16 \tau}{\e} \right\}\right);
	\end{align*}
	setting $\tau = (\rho + \sigma) \sqrt{(s/n)}$ makes the right-hand side at most $2 \delta$ (this uses the fact that $\rho \ge 1$ and $s \ge \e^2 n$). Since $\e' \frac{s}{n} (\rho + \sigma) = \e \tau$, this concludes the proof of the lemma.

\section{Proof of Theorem \ref{thm:one-phasePC}} \label{app:otl}

	First we verify that with good probability we have $\estopt \in [\frac{1}{2}\, \OPT(\bL^{>\frac{n}{2}}(1-\e')), \OPT(\bL^{>\frac{n}{2}}(1-\e'))]$ so we can employ Theorem~\ref{thm:LPAssEstimates} to obtain guarantees for the solution $\{\bar{\bx}^t\}_{t=n/2+1}^n$ returned by \LPtoLB. 

	Taking a union bound over Parts (b) and (c) of Lemma~\ref{lemma:PCLB} and using the fact that $\sigma \ge 1$, we have that
\[ \Pr \left[ \left(1 - \frac{\e'}{2}\right) \le \frac{\OPT(\bL^{\le \frac{n}{2}}(1 - \e'))}{\frac12 \OPT(\bL)} \le  (1 + 3 \e' \sigma) \right] \geq 1 - 6\delta, \]
and the same holds for $\OPT(\bL^{>\frac{n}{2}}(1-\e'))$; let $\cE$ be the event that these bounds hold for both, so we have $\Pr(\cE) \ge 1 - 12\delta$. We can chain the bounds for $\OPT(\bL^{\le \frac{n}{2}}(1-\e'))$ and $\OPT(\bL^{>\frac{n}{2}}(1-\e'))$ to get that conditioned on $\cE$
	\begin{align*}
		\frac{(1-\frac{\e'}{2})^2}{(1+3 \e' \sigma)^2} \cdot \OPT(\bL^{>\frac{n}{2}}(1-\e')) \le \estopt \le \OPT(\bL^{>\frac{n}{2}}(1-\e')).
	\end{align*}
	Now using the fact that $\epsilon \in (0,1)$ and $\sigma \ge 1$, it is easy to verify that we can lower bound the first factor as $\frac{(1-\frac{\e'}{2})^2}{(1+3 \e' \sigma)^2} \ge 1 - 10\e' \sigma$. In addition, since $\sigma \le \frac{1}{20 \e'}$, this in particular implies $\estopt \ge \frac{1}{2} \OPT(\bL^{>\frac{n}{2}}(1-\e'))$.

	Then under the event $\cE$ the requirements of the Theorem~\ref{thm:LPAssEstimates} holds and thus taking a union bound over the event $\cE$ and this theorem gives that with probability at least $1-13\delta$ the solution $\{\bar{\bx}^t\}_{t=n/2+1}^n$ is $c_2 \e$-feasible for $\bL^{> \frac{n}{2}}(1-\e')$ (and thus is $(c_2 \e + \e')$-feasible for $\bL^{> \frac{n}{2}}$) and has value 
	\begin{align*}
	\sum_{t = n/2 + 1}^n \bpi^t \bar{\bx}^t &\ge (1-c_3 \e)\, \estopt \ge (1-c_3 \e) (1-10\e' \sigma)\, \OPT(\bL^{> \frac{n}{2}}(1-\e'))\\
	 	&\ge (1-c_5 \e \sigma)\, \OPT(\bL^{> \frac{n}{2}}(1-\e')),
	\end{align*}
	for some constant $c_5$, where the last inequality uses the fact $\sigma \ge 1$, which gives the desired guarantees. (Actually, in more detail notice that the event $\cE$ only depends on what \emph{set} of blocks $(\pi^t, A^t, C^t)$ of the underlying LP $\L$ comprise the blocks in the random sampled LP $\bL^{> \frac{n}{2}}(1-\e')$ (this also determines the blocks that comprise the sampled LP $\bL^{\le \frac{n}{2}}(1-\e')$). Thus, we can first condition that the set of blocks in $\bL^{>\frac{n}{2}}(1-\e')$ is a specific one that satisfies $\cE$, and since we do not fix their \emph{order} they still appear in random order in the conditioned $\bL^{>\frac{n}{2}}(1-\e')$ and we can apply Theorem~\ref{thm:LPAssEstimates} to this conditioned $\bL^{>\frac{n}{2}}(1-\e')$; ranging over all such conditioning then gives the result.) 


\section{Proof of Lemma \ref{lemma:phaseDLA}} \label{app:phaseDLA}

	First we claim that with probability at least $1-6\delta$, the random LP $\bL^{S_i}(1-\e_i)$ is $(\e_i, 7\sigma)$-stable. Notice that we can employ Lemma \ref{lemma:PCLB}(c) to $\L$ with parameter $\e'$ set to $2\e_i$ (in particular, $2 \e_i \sqrt{|S_i|/n} \ge \e$ plus Lemma \ref{lemma:basicStab} implies that $\L$ is $(2 \e_i \sqrt{|S_i|/n}, \sigma)$-stable). Thus, Lemma \ref{lemma:PCLB}(c) guarantees that with probability at least $1 - 2\delta$, $\OPT(\bL^{S_i}(1 - 2 \e_i)) \le (1 + 6 \sigma \e_i) \frac{|S_i|}{n} \OPT(\L)$. Moreover, applying Lemma \ref{lemma:PCLB}(b) to $\L$ with $\e' = \e_i$ we get that with probability at least $1 - 4\delta$, $\OPT(\bL^{S_i}(1 - \e_i)) \ge (1 - \frac{\e_i}{2}) \frac{|S_i|}{n} \OPT(\L)$. Putting these inequalities together, we get that with probability at least $1-6\delta$
	\begin{align*}
		\OPT(\bL^{S_i}(1 - 2\e_i)) \le \frac{(1 + 6 \sigma \e_i)}{(1 - \frac{\e_i}{2})} \OPT(\bL^{S_i}(1 - \e_i)) \le (1 + 7 \sigma \e_i) \OPT(\bL^{S_i}(1 - \e_i)),
	\end{align*}
	where the last inequality uses $\sigma \ge 1$ and that $\e_i \le \sqrt{\e}$ is at most a sufficiently small constant (for example $\sqrt{\e} \le \frac{1}{7}$ suffices). This proves the claim.

	Second, notice that whenever the claim above holds (i.e. $\bL^{S_i}(1-\e_i)$ is $(\e_i, 7\sigma)$-stable) we can apply Theorem \ref{thm:one-phasePC} with parameters $(\e_i, \delta, 7\sigma)$ to obtain guarantees for the solution $(0, \ldots, 0, \bar{\bx}^{|S_{i-1}|+1}, \ldots, \bar{\bx}^{|S_i|})$. So employing the claim above, Theorem \ref{thm:one-phasePC} and union bound, we get that with probability at least $1-19\delta$ this solution is $(c_4 \e_i)$-feasible for $\bL^{S_i \setminus S_{i-1}}(1-\e_i)$ (and thus $((c_4 + 1) \e_i)$-feasible for $\bL^{S_i \setminus S_{i-1}}$) and has value $\sum_{t = |S_{i-1}|+1}^{|S_i|} \bpi \bar{\bx}^t \ge (1-c_5 7 \sigma \e_i) \OPT(\bL^{S_i \setminus S_{i-1}}(1-\sqrt{2} \e_i))$. This concludes the proof of the lemma. (As noted at the end of Theorem \ref{thm:one-phasePC}, formally we should condition on the \emph{set} of blocks of $\bL^{S_i}$ being such that the claim above holds and then apply Theorem \ref{thm:one-phasePC} on the conditioned LP (which still has its blocks coming in random order), and then take an average over all such conditionings.)
 \section{Proof of Lemma \ref{lem:est-rhs}} \label{app:estRhsOpt}

	The following will be a more convenient way of seeing the random order process: the random restricted LP $\bL^{I}$ (i.e. columns are in random order and we are focusing on the ones indexed by $I$) has the same distribution as the LP $\L^{\bI}$ for $\bI$ a uniformly chosen $|I|$-subset of $[n]$ (i.e. sampling $|I|$-columns from $\L$ without replacement).
	
	The following lemma is the main part of the argument.

\begin{lemma}
  \label{lem:est-rhsPre}
  Let $\e \leq 1/20$. Given $\L \in \PLPn$ with width at least $40 B$, let $\tau$ be an $\OPT(\L)$-good skimming threshold for $\L$. Let $\bK$ be a random $pn$-subset of $[n]$ for $p \in [\frac14, \frac34]$. Then with probability at least $1 - 2\delta$,
  \begin{align*}
    \RHS(\skim{\L^\bK}{\tau}) \in (1 \pm \e)\,p\, \RHS(\skimLt)
  \end{align*}
  and with probability at least $1-8\delta$,
  \begin{align*}
    \OPT(\skim{\L^{\bK}}{\tau}) \in (1 \pm \e) p \,\OPT(\skim{\L}{\tau}) \pm \e 64 p B \tau.
  \end{align*}
\end{lemma}

\begin{proof}
  Let $H = \{t : \pi_t > \tau\}$ be the set of
  high-value items. If $A^t$ is the $t^{th}$ column of $A$, we get that 
  \begin{align*}
    b_L &:= \RHS(\skimLt) ~~=~~ \textstyle \RHS(\L) - \sum_{t \in H}
    A^t~~; \\
    b^{\bK}_L &:= \RHS(\skim{\L^\bK}{\tau}) ~~=~~ \textstyle \RHS(\L^\bK) - \sum_{t \in
      \bK: \pi_t > \tau} A^t ~~=~~ p\, \RHS(\L) - \sum_{t
      \in H \cap \bK} A^t.
  \end{align*}
  Since $\bK$ is a random $pn$-subset of the columns, the expectation
  of the second line is precisely $p$ times the first line. To obtain concentration, fix a coordinate $i$. 
The second property of a good skimming threshold implies the high-value items occupy at most half of the budget: $\sum_{t \in H} A^t_i \le \frac{\RHS(\L)_i}{2} \le (b_L)_i$. This gives $\E\big[\sum_{t \in H \cap \bK} A^t_i\big] = p \,\sum_{t \in H} A^t_i \le p\, (b_L)_i$ and Bernstein's inequality Corollary \ref{cor:multiChernoff} gives that $\sum_{t \in H \cap \bK} A^t_i = p\,\sum_{t \in H} A^t_i \pm \e\,p\, (b_L)_i$ with probability at least
  \begin{align*}
  	1 - 2\exp\left\{ - \frac{(\e\,p\, (b_L)_i)^2}{(4 + \e) p\,(b_L)_i} \right\} \geq 1
  - 2\exp\left\{ - \frac{\e^2 p\,(b_L)_i}{5} \right\} \geq 1 - \frac{2\delta}{m},
  \end{align*}
 where the last inequality uses the width lower bound $(b_L)_i \ge \frac{1}{2} \, \RHS(\L)_i \ge 20 B$ (recall we have used the normalization $A_i^t \in [0,1]$).
 A union bound over all coordinates $i$ then gives that $b^{\bK}_L \in p\, b_L (1 \pm \e)$ with probability at least $1 - 2 \delta$, giving the first part of the lemma.
%

  For the estimate of the optimal value, let $L := [n] \setminus H$ be the low-value items. Let $\skim{\L}{\tau}\big(\frac{1}{p} \cdot b^{\bK}_L\big)$ be the LP $\skim{\L}{\tau}$ with the right-hand side replaced by $\frac{1}{p} \cdot b^{\bK}_L$ and notice that
  \begin{align}
  	\textstyle  \skim{\L^\bK}{\tau} = \max \big\{ \sum_{t \in L \cap \bK}
  \pi_t x_t \mid \sum_{t \in \bK} A^t x_t \leq b^\bK_L \big\} = \big(\skim{\L}{\tau}\big(\frac{1}{p} \cdot b^{\bK}_L\big)\big)^{\bK}. \label{eq:paramPacking1}
  \end{align} 
  The idea is to apply Lemma \ref{lemma:PCLB} to bound $\skim{\L}{\tau}\big(\frac{1}{p} \cdot b^{\bK}_L\big)$, but unfortunately the right-hand side of this LP is random; but since from the previous part $\frac{1}{p} \cdot b^{\bK}_L \in b_L (1 \pm \e)$ with good probability, we replace the right-hand side with the upper bound $(1+\e) b_L$ (resp. lower bound $(1-\e) b_L$) to upper (resp. lower) bound the optimal value. So define $\L_+ = \skim{\L}{\tau}\big(b_L (1+\e)\big)$ and employ Lemma \ref{lemma:PCLB}(c) to $\L_+$ with $\rho = \frac{32 \tau B}{\OPT(\L_+)}$ to get that with probability at least $1-2\delta$ 
  \begin{align}
  	\OPT(\L_+^{\bK}) \le \left(1 + \frac{\e 64 \tau B}{\OPT(\L_+)}\right) p \,\OPT(\L_+) \le (1 + \e) p\, \OPT(\skim{\L}{\tau}) + \e 64 p B \tau, \label{eq:paramPacking2}
  \end{align}
  where again $\OPT(\L_+) \le (1+\e)\, \OPT(\skim{\L}{\tau})$ since an optimal solution to $\L_+$ can be scaled by $\frac{1}{1+\e}$ to give a feasible solution to $\skim{\L}{\tau}$. Moreover, from the first part of this lemma, with probability at least $1-2\delta$ we have $\frac{1}{p} \cdot b_L^{\bK} \le b_L (1 + \e)$, in which case $\OPT(\skim{\L}{\tau}(\frac{1}{p} \cdot b_L^{\bK})^{\bK}) \le \OPT(\L_+^{\bK})$. Thus, together with inequalities \eqref{eq:paramPacking1} and \eqref{eq:paramPacking2}, a union bound gives that with probability at least $1-4\delta$  %
  \begin{align*}
  	\OPT(\skim{\L^{\bK}}{\tau}) \le  (1 + \e) p\, \OPT(\skim{\L}{\tau}) + \e 64 p B \tau.	
  \end{align*}
  
  Using the same reasoning, the lower bound 
  \begin{align*}
  	\OPT(\skim{\L^{\bK}}{\tau}) \ge (1 - \e) p \,\OPT(\skim{\L}{\tau}) - \e 64 p B \tau
  \end{align*}
  holds with probability at least $1-4\delta$. Taking a union bound over these upper and lower bounds then concludes the proof of the lemma. 
\end{proof}


\begin{proof}[Proof of Lemma \ref{lem:est-rhs}]
	Let $\bI$ be a random $\frac{n}{3}$-subset of $[n]$ and let $\bJ = [n] \setminus \bI$, and recall $\skim{\bL^{\le \frac{n}{3}}}{\tau}$ has the same distribution as $\skim{\L^{\bI}}{\tau}$ and $\skim{\bL^{> \frac{n}{3}}}{\tau}$ has the same distribution as $\skim{\L^{\bJ}}{\tau}$. Thus, employing Lemma \ref{lem:est-rhsPre} with $\bK = \bI$ and then $\bK = \bJ$ and employing a union bound, we get that with probability at least $1-4\delta$
	\begin{align*}
    \RHS\big(\skim{\bL^{\le \frac{n}{3}}}{\tau}\big) \in \left[\frac{1-\e}{1+\e} , \frac{1+\e}{1-\e} \right] \frac{1}{2}\, \RHS\big(\skim{\bL^{> \frac{n}{3}}}{\tau}\big) \subseteq (1 \pm 3 \e) \,\frac{1}{2}\, \RHS\big(\skim{\bL^{> \frac{n}{3}}}{\tau}\big),
	\end{align*}
	where the last inequality uses $\e \le \frac{1}{3}$.

	The bounds for $\OPT\big(\skim{\bL^{\le \frac{n}{3}}}{\tau}\big)$ follow similarly. This concludes the proof. 
\end{proof}
 \section{Proof of Theorem \ref{thm:PLP-main}} \label{app:mDLA}
	
	We prove only the first part of the theorem since the second directly follows from it. 
Let $\cE$ be the event that the the guarantee of Theorem \ref{thm:one-phase} holds for all $i = 1, \ldots, \log \frac{1}{\e}$. Since each such guarantee holds with probability at least $1-15\delta$, a union bound gives that $\cE$ holds with probability at least $1 - 15 \delta \log \frac{1}{\e}$.	As in the proof of Theorem \ref{thm:multi-phasePC}, it is easy to see that under $\cE$ the returned solution $\{\bar{\bx}_t\}_t$ is feasible for $\bL$. 

Now we analyze the value $\val$ of the solution. To simplify the notation, let $\OPT(I) = \OPT(\bL^I)$ for any subset $I \subseteq [n]$.
Again, conditioned on $\cE$ we can put together the guarantees from Theorem \ref{thm:one-phase} for each phase of \mDLA to get (recall in phase $i$ we use the approximation parameter $\e_i = \e \sqrt{n/|S_i|}$)
	\begin{align}
		\E[\val \smid \cE] \ge \sum_{i = 1}^{\log \e^{-1}} \E\big[\OPT(S_i \setminus S_{i-1}) - c_6 \e_i \OPT(S_i) \smid \cE\big]. \label{eq:mDLA1}
	\end{align}	
	
	First we analyze the unconditioned value of the right-hand side of this expression. 
	
	\begin{lemma} \label{lemma:proofMDLA}
		$\sum_{i = 1}^{\log \e^{-1}} \E[\OPT(S_i \setminus S_{i-1})] - c_6 \sum_{i = 1}^{\log \e^{-1}} \e_i\, \E[\OPT(S_i)] \ge (1 - 2\e - 2 \delta \log \e^{-1} - 4c_6 \e)\, \OPT(\L)$.
	\end{lemma}
	
	\begin{proof}
		The last term in the left-hand side is at most $4 c_6 \e\, \OPT(\L)$: Lemma \ref{lemma:PCLB}(d) gives $\E[\OPT(S_i)] \le \frac{|S_i|}{n} \OPT(\L)$ (since our LP without covering constraints is trivially $(\e', 0)$-stable), so we get
	\begin{align*}
		 c_6 \sum_{i = 1}^{\log \e^{-1}} \e_i\, \E[\OPT(S_i)] \le c_6 \sum_{i = 1}^{\log \e^{-1}} \e_i\, \frac{|S_i|}{n} \, \OPT(\L) = c_6 \sum_{i = 1}^{\log \e^{-1}} \e_i\, \e\, 2^i\, \OPT(\L) \le 4 c_6 \e\, \OPT(\L),
	\end{align*}
	where the last inequality follows from the bound $\sum_{i=1}^{\log \e^{-1}} \e_i 2^{i-1} \le 2$ from equation \eqref{eq:epsSum}.
	
	To show that the first term in the left-hand side is approximately $\OPT(\L)$ is a little more complicated, since without good bounds on the ratios $\frac{\pi_t}{\OPT(\L)}$ we do not have bounds showing that $\OPT(S_i \setminus S_{i-1}) \approx \frac{|S_i| \setminus |S_{i-1}|}{n}\, \OPT(\L)$ even in expectation; we need to look at all the terms $\OPT(S_i \setminus S_{i-1})$ simultaneously.

	Let $\bx^*$ be a optimal solution for $\bL$ and let $\bx^*(i)$ be the solution scaled down by $(1 - \frac{\e_i}{2})$ and restricted to the $2^{i-1} \e
n$ columns in $S_i \setminus S_{i-1}$. From Lemma \ref{lemma:PCLB}(a) we have that $\bx^*(i)$ is feasible for $\bL^{S_i \setminus S_{i-1}}$ with probability at least $1 - 2\delta$. Let $\mathcal{G}$ be the good event that for all $i$ we have $\bx^*(i)$ feasible for $\bL^{S_i \setminus S_{i-1}}$, which holds with probability at least $1 - 2 (\log \e^{-1}) \delta$ by a union bound. Then for every scenario in $\mathcal{G}$ we have $\OPT(S_i \setminus S_{i-1}) \ge (1 - \frac{\e_i}{2}) \sum_{t \in S_i \setminus S_{i-1}} \bpi_t \bx^*_t$, so adding over all $i \ge 1$ gives 
\[ \sum_{i \geq 1} \OPT(S_i \setminus S_{i-1}) \geq \OPT(\L) - \sum_{i \geq 1} \frac{\e_i}{2}
\sum_{t \in S_i \setminus S_{i-1}} \bpi_t \bx^*_t - \sum_{t \in S_0} \bpi_t \bx^*_t. \]
Then letting $\ones_{\cG}$ be the indicator of the event $\cG$ and using the non-negativity of the objective function, we get that
	\begin{align*}
		\E\bigg[ \sum_{i \geq 1} \OPT(S_i \setminus S_{i-1}) \bigg] &\ge \E\bigg[\ones_{\cG} \cdot \sum_{i \geq 1} \OPT(S_i \setminus S_{i-1}) \bigg] \ge \E[\ones_{\cG} \cdot \OPT(\L)] - \sum_{i \geq 1} \frac{\e_i}{2} \sum_{t \in S_i \setminus S_{i-1}} \E[\bpi_t \bx^*_t] - \sum_{t \in S_0} \E[\bpi_t \bx^*_t].
	\end{align*}
	Clearly $\E[\ones_{\cG} \cdot \OPT(\L)] = \Pr(\cG) \cdot \OPT(\L) \ge (1 - 2 \delta \log \e^{-1})\, \OPT(\L)$. Moreover, by the random order model we have that for all $t$, $\E[\bpi_t \bx^*_t] = \frac{\OPT(\L)}{n}$. Thus, the above bound becomes 
	\begin{align*}
		\E\bigg[ \sum_{i \geq 1} \OPT(S_i \setminus S_{i-1}) \bigg] \ge (1 - 2\delta \log \e^{-1})\, \OPT(\L) - \frac{\OPT(\L)}{n} \sum_{i \ge 1} \frac{\e_i \, \e\, 2^{i-1} n}{2} - \e \OPT(\L) \ge (1 - 2\e - 2\delta \log \e^{-1})\, \OPT(\L),
	\end{align*}
	where again we have used equation \eqref{eq:epsSum} in the last step. This concludes the proof of the lemma. 
\end{proof}

	To handle the \emph{conditioned} right-hand side of \eqref{eq:mDLA1} we observe the following: with probability 1 we have $\OPT(S_0) + \sum_{i=1}^{\log \e^{-1}} \OPT(S_i \setminus S_{i-1}) \le \OPT(\L)$, since concatenating the optimal solutions for $\bL^{S_0}$, $\bL^{S_1 \setminus S_0}$, \ldots, gives a feasible solution for $\bL$. For convenience, let $\val_i := \OPT(S_i \setminus S_{i-1}) - c_6 \e_i \OPT(S_i)$. Then using the decomposition $\E[\sum_{i \ge 1} \val_i] = \E[\sum_{i \ge 1} \val_i \smid \cE] \Pr(\cE) + \E[\sum_{i \ge 1} \val_i \smid \bar{\cE}] \Pr(\bar{\cE})$, inequality \eqref{eq:mDLA1} gives
	\begin{align*}
		\E[\val \smid \cE] &\ge \sum_{i \ge 1} \E[\val_i \smid \cE] \ge \sum_{i \ge 1} \E[\val_i] - \Pr(\bar{\cE}) \sum_{i \ge 1} \E[\val_i \smid \bar{\cE}] \\
		&\ge \sum_{i \ge 1} \E[\val_i] - \Pr(\bar{\cE}) \Bigg[\E[\OPT(S_0) \smid \bar{\cE}] + \sum_{i \ge 1} \E[\OPT(S_i \setminus S_{i-1}) \smid \bar{\cE}] \Bigg] \\
		& \ge \sum_{i \ge 1} \E[\val_i] - \Pr(\bar{\cE})\, \OPT(\L).
	\end{align*}
	Lower bounding the first term using Lemma \ref{lemma:proofMDLA} and recalling that $\Pr(\bar{\cE}) \le 15\delta \log \e^{-1}$, gives $\E[\val \smid \cE] \ge (1-2\e-17\delta\log \e^{-1} - 4 c_6\e) \,\OPT(\L)$. This concludes the proof of Theorem \ref{thm:PLP-main}.
\fi

\end{document}